\documentclass[journal]{IEEEtran}
  \usepackage[pdftex]{graphicx}
\usepackage[cmex10]{amsmath}
\usepackage{amsthm}
\theoremstyle{plain}
\newtheorem{theorem}{Theorem}
\newtheorem{lemma}{Lemma}
\begin{document}
\title{On the Performance of Adaptive Packetized Wireless Communication Links under Jamming}
\author{Koorosh~Firouzbakht,
        Guevara~Noubir,
        Masoud~Salehi% <-this % stops a space
\thanks{K.~Firouzbakht and M.~Salehi (Emails: \{kooroshf, salehi\}@ece.neu.edu) are with the Electrical Engineering Department, Northeastern University, Boston, MA 02115.}%
\thanks{G.~Noubir (Email: noubir@ccs.neu.edu) is with the College of Computer and Information Science, Northeastern University, Boston, MA 02115. }%
\thanks{Research partially supported by NSF Award CNS-0915985.}% <-this % stops a space
\thanks{This paper was partially presented at the ACM WiSec 2012 Conference, April 16-18, 2012 Tucson, Arizona, USA}%
\thanks{This work has been submitted to the IEEE for possible publication.  Copyright may be transferred without notice, after which this version may no longer be accessible.}%
}%
% *** The paper running headers ***
\markboth{Submitted to IEEE Transactions on Wireless Communications}%
{Shell \MakeLowercase{\textit{K.~Firouzbakht, G.~Noubir and M.~Salehi}}: Paper Title}
\maketitle
\begin{abstract}
\boldmath
We employ a game theoretic approach to formulate communication between two nodes over a wireless link in the presence of an adversary. We define a constrained, two-player, zero-sum game between a transmitter/receiver pair 
with adaptive transmission parameters
% with multiple communication choices, 
and an adversary with average and maximum power constraints. In this model, the \emph{transmitter}'s goal is to maximize the achievable expected performance of the communication link, defined by a utility function, while the \emph{jammer}'s goal is to minimize the same utility function. Inspired by capacity/rate as a performance measure, we define a general utility function and a payoff matrix which may be applied to a variety of jamming problems. We show the existence of a threshold $(J_{\text{TH}})$ such that if the jammer's average power exceeds $J_{\text{TH}}$, the expected payoff of the transmitter at Nash Equilibrium (NE) is the same as the case when the jammer uses its maximum allowable power, $J_{\max}$, all the time. We provide analytical and numerical results for transmitter and  jammer optimal strategies and a closed form expression for the expected value of the game at the NE. As a special case, we investigate the maximum achievable transmission rate of a rate-adaptive, packetized, wireless AWGN communication link under different jamming scenarios and show that randomization can significantly assist a smart jammer with limited average power.  
\end{abstract}

\IEEEpeerreviewmaketitle
%
% 
% ******************************************************************************************************************
% 
\section{Introduction}
\label{Sec:Introduction}
\IEEEPARstart{O}{ver} the last decades, wireless communication has been established as an enabling technology to an
increasingly large number of applications. The convenience of wireless and its support of mobility
has revolutionized the way we access information services and interact with the physical
world. Beyond enabling mobile devices to access information and data services ubiquitously,
wireless technology is  widely used in cyber-physical systems such as air-traffic control, power
plants synchronization, transportation systems, and human body implantable devices. This
pervasiveness has elevated wireless communication systems to the level of critical infrastructure.
Radio-frequency wireless communications occur over a broadcast medium, that is not only shared
between the communicating nodes but is also exposed to adversaries. Jamming is one of the most
prominent security threats as it not only can lead to denial of service attacks, but can also be
a prelude to spoofing attacks.

Anti-jamming has been an active area of research for decades. Various techniques for combating
jamming have been developed at the physical layer
% \cite{SimonOSL01}
which include directional
antennas, spread spectrum communication and power/modula\break tion/coding control. At the time, most of wireless systems were neither packetized nor networked. Reliable communication in the presence of adversaries has regained significant interest in the last few years, as new jamming attacks as well as need for more complex applications and deployment environments have emerged. Several specifically crafted attacks and counter-attacks have been proposed for packetized wireless data networks~\cite{negi+p:jam,LinN04,LiKP07,WilhelmMSL11}, 
multiple access resolution~\cite{BenderFHKL05,GilbertRC06,BayraktarogluKLNRT08, AwerbuchRS08}, 
multi-hop networks~\cite{TagueSNP08, LiKP07}, 
broadcast and control communication~\cite{KBKV06, ChiangH07, ChanLNT07, TagueLP07,LazosLK09,LiuPDL10,LiuLK11}, cross-layer resiliency~\cite{LinN05}, wireless sensor networks~\cite{XuKWY06,XuTZ08}, spread-spectrum without shared secrets~\cite{StrasserCSM08, SlaterPRB09, JinNT09}, and navigation information broadcast systems~\cite{RasmussenCC07}. 

Nevertheless, very little work has been done on protecting rate adaptation algorithms against 
adversarial attacks. Rate adaptation plays an important role in widely used wireless 
communication systems such as the IEEE 802.11 standard as the link quality in a WLAN is often 
highly dynamic. In recent years, a number of algorithms for rate adaptation have been proposed in the literature~\cite{HollandVB01,JuddXP08,VutukuruHK09,RahulFDC09,RamachandranKZG08,CampE08,KimSSD06,WongHSV06}, and some are widely deployed~\cite{Bicket05,LacageMT04}. 
Recently, rate adaptation for the widely used
IEEE 802.11 protocol was investigated in~\cite{PelechrinisBKG09, NoubirRST11}. Experimental and
theoretical analysis of optimal jamming strategies against currently deployed rate adaptation
algorithms indicate that the performance of IEEE 802.11 can be significantly degraded with very few interfering
pulses. The commoditization of software radios makes these attacks very practical and calls for
investigation of the capacity of packetized communication under adaptive \nolinebreak jamming.

In this work, we focus on the problem of determining the optimal transmission strategies and adaptation mechanisms for a transmitter/receiver with multiple transmission choices/ parameters 
% \footnote{We will use the terms \emph{choices} and \emph{parameters} interchangeably throughout the paper.} 
(multiple transmission rates, different transmission powers, etc) when the wireless channel is subject to jamming by a power constrained jammer. We consider a setup where a pair of nodes (transmitter and receiver) communicate using data packets. An adversary can interfere with the communication but is constrained by an instantaneous maximum power per packet ($J_{\max}$) as well as a long-run average power ($J_{\text{ave}}$).

Packets each selected with appropriate transmission parameters, either overcome the interference or are lost otherwise. Inappropriate selection of the transmission parameters can either increase the chance of underperforming (if the parameters are selected conservatively) or loosing a packet (if the parameters are selected aggressively). In this communication scenario, it is crucial to understand the interaction between the communicating nodes and the adversary, determine the long-term achievable maximum performance and the optimal transmitter strategy to achieve it, as well as the optimal strategy for the adversary. While, for a channel with fixed-power jammer, the optimal strategies for communication and jamming and the system performance are derived from the fundamental information theoretic results (See Section~\ref{Sec:Game_Analysis}), these questions are still open for a packetized communication~system. 

Our contributions can be summarized as follows:
\begin{itemize}
\item We formulate the interaction between the communicating nodes and an adversary using a
game-theoretic framework. We show the existence of the Nash Equilibrium (NE) for this non-typical constrained zero-sum game. 
% We also show that the NE strategies can be computed using linear programming.
\item We show that the standard information-theoretic results for a jammed
channel correspond to a pure NE.
\item We further characterize the game by showing that, when both players are allowed to
randomize their actions (e.g., coding rate and jamming power) a new NE appears with
surprising properties. We show the existence of a jamming threshold ($J_{\text{TH}}$) such that if the jammer
average power exceeds $J_{\text{TH} }$, the game value at the NE is the same as the case when the jammer uses $J_{\max}$ all \nolinebreak the \nolinebreak time. 
\item We provide analytical results for the optimal NE strategies and the expected value of the game at NE as a function of jammer's average power.
\end{itemize}
% {\bf Paper Structure:} 
The remainder of the paper is organized as follows: In Section \ref{Sec:System_Model}, we introduce and define our model for the communication link in an adversarial setting. In Section \ref{Sec:Game_Model}, we formulate the interactions between the communicating nodes and the adversary as a constrained two-player zero-sum game and define a general utility function and a payoff matrix which are applicable to a variety of jamming problems.  Additionally, we discuss how the additional constraint on jammer's average power makes our game model different from a typical zero-sum game. In Section~\ref{Sec:Game_Characterization}, we show the existence of the NE for our constrained zero-sum game. We also prove the existence of the jamming threshold and its effect on the game outcome. In Section~\ref{Sec:Game_Analysis}, we provide analytical results for the players' optimal strategies and the game value at the NE when the jammer's average power takes different values. In section \ref{Sec:Special_Case} we derive analytical and numerical results for two special cases when the utility functions and payoffs are defined as the AWGN channel capacity. Finally, we conclude the paper in Section~\ref{Sec:Conclusion}.
%
% 
% ******************************************************************************************************************
% 
\begin{figure}
 \centering
 \includegraphics[width = 3in]{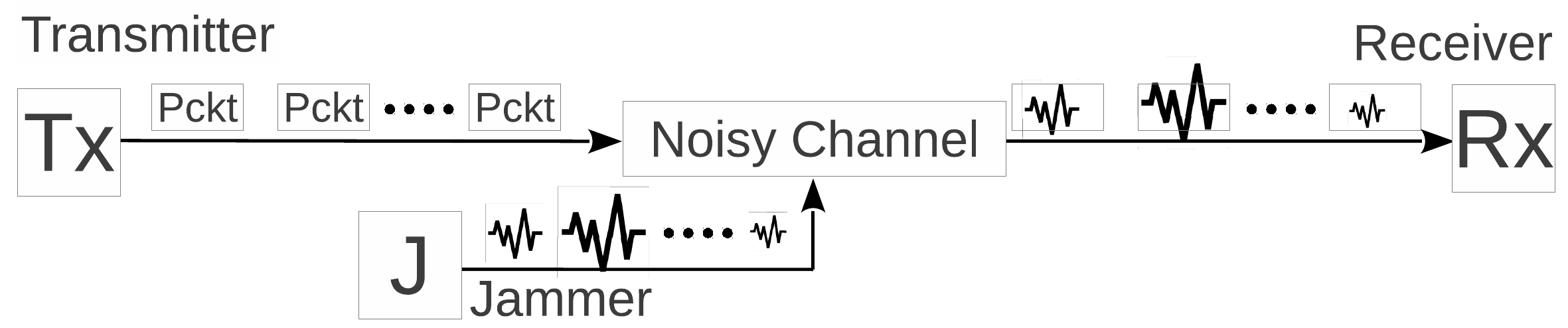}
 \caption{The system model; packetized communication link under jamming}
 \label{Fig:System_Model}
\end{figure}
%
%
% 
% ******************************************************************************************************************
% 
\section{System Model and Problem Statement}
\label{Sec:System_Model}
In this section we formally define the problem under study. The corresponding system model is shown in Figure \ref{Fig:System_Model}. The transmitter and the receiver are communicating through a packetized, wireless noisy channel. Beside the channel noise, the transmitted packets are also disrupted by an adversary, the \emph{jammer}. The jammer's maximum and average jamming powers are assumed to be limited to $J_{\max}$ and $J_{\text{ave}}$, respectively.
%
%***************************************************************************************************
%
% 
\subsection{The Channel Model}
\label{SubSec:Channel_Model}
The wireless communication link between the transmitter and the receiver is assumed to be a single-hop, noisy channel
% (not necessarily additive as depicted in Figure \ref{Fig:System_Model})  
with fixed and known channel parameters. Furthermore, the communication link is being disrupted by an adversary
% (not necessarily additive as depicted)
, the \emph{jammer}. The jammer transmits radio signals to increase the effective noise at the receiver and hence degrades the performance of the communication link (e.g., to decrease the channel capacity or throughput, degrade the quality of service, etc.) between the transmitter and the receiver. 
We assume \emph{packet-based} transmission, i.e., transmissions occur in disjoint time intervals (time slots) during which transmitter's and jammer's state (parameters) remain unchanged. 

In Section~\ref{Sec:Game_Model} we introduce and study a constrained two-player zero-sum game between the transmitter-receiver pair and the jammer in which the goal of the transmitter-receiver pair is to achieve the highest performance (e.g., channel capacity, channel throughput, etc.) while the jammer tries to minimize the achievable performance.
%
%***************************************************************************************************
%
\subsection{The Transmitter Model}
\label{SubSec:Transmitter_Model}
\begin{figure}
 \centering
 \includegraphics[width = 3in]{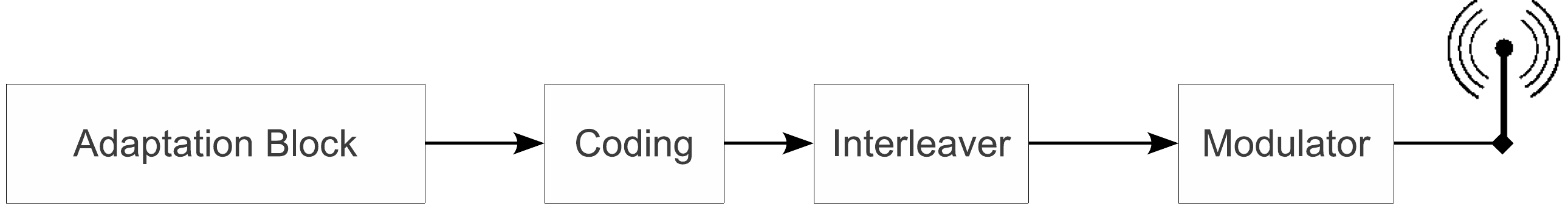}
 \caption{Transmitter model}
 \label{Fig:Transmitter_Model}
\end{figure}
The transmitter has an adaptation block which enables him to change and adapt his transmission parameters (e.g., transmission power, rate, modulation, etc.). 
In order to combat jamming, the transmitter changes his transmission parameters according to a probability distribution (his strategy). The transmitter chooses an optimal distribution to achieve the best average performance (or his expected payoff) which is presented by a preference/utility function. Common measures of performance in wireless networks are achievable capacity, network throughput, quality of service (QoS), power consumption, etc. \cite{Manshaei2013}.
As stated before, we assume transmissions are \emph{packet-based}.
The transmitter's model is shown in Figure~\ref{Fig:Transmitter_Model}.

The interleaver block in transmitter's model is a countermeasure to burst errors and burst jamming (transmitting a burst of white noise to disrupt a few bits in a packet). Interleaving is frequently used in digital communications and storage devices to improve the burst error correcting capabilities of a code. Burst errors are specially troublesome in short length codes as they have limited error correcting capabilities. In such codes, a few number of errors could result in a decoding failure or an incorrect decoding. A few incorrectly decoded codewords within a larger frame could make the entire frame corrupted.

Combining effective interleaving schemes such as cryptographic interleaving and capacity-achieving codes, such as turbo and LDPC codes, results in effective transmission schemes 
% that have good burst error correcting properties 
(see \cite{LinN04}) which make burst jamming ineffective.
Therefore, in our study we do not consider burst jamming.
% and assume the jammer remains active during the entire packet transmission.
% 
%
%***************************************************************************************************
%
% 
\subsection{The Jammer Model}
\label{SubSec:Jammer_Model}
Radio jamming or simply \emph{jamming} is deliberate transmission of radio signals with the
intention of degrading performance of a communication link. 
A fairly large number of jamming models have been proposed in the literature \cite{Peterson}. The
most benign jammer is the \emph{barrage noise jammer}.  The barrage noise jammer transmits
bandlimited white Gaussian noise with power spectral density (psd) of $J$. It is usually assumed
that the barrage noise jammer power spectrum covers exactly the same frequency range as the
communicating system. This kind of jammer simply increases the Gaussian noise level from $N$ to
$(N+J)$ at the receiver's front end. Another frequently used jamming model is the \emph{pulse-noise
jammer}. The pulse noise jammer transmits pulses of bandlimited white Gaussian noise having total
average power of $J_{\text{ave}}$ referred to the receiver's front end. It is usually assumed that the
jammer chooses the center frequency and bandwidth of the noise to be the same as the transmitter's
center frequency and bandwidth. The jammer chooses its pulse duty factor to cause maximum
degradation to the communication link while maintaining the average jamming power $J_{\text{ave}}$. For a
more realistic model, the pulse-noise jammer could be subject to a maximum peak power constraint.
Other jamming models, to name a few, are  the \emph{partial-band jammer} and
\emph{single/multiple-tune jammer}.

We study a more sophisticated jamming model. The jammer in study is a reactive jammer, i.e., he is only active when a packet is being transmitted and silent otherwise.
We assume that the jammer uses a set of discrete jamming power levels arbitrary placed between
$J=0$ and $J=J_{\max}$. The jammer can choose any jamming power level to increase the effective noise at the receiver, but he has to maintain an
overall average jamming power, denoted by $J_{\text{ave}}$. The jammer uses  his available power levels according to
a probability distribution (his strategy), he chooses an optimal strategy to minimize the performance of the communication link while maintaining his maximum and average power constraints.
%, $J_{\max}$ and $J_{\text{ave}}$, respectively.
As shown in Section \ref{SubSec:Transmitter_Model}, burst jamming is not an optimal jamming scheme and hence, we assume that the jammer remains active during the entire packet transmission, i.e., the jammer transmits a
continuous jamming signal with a fixed power (variance) $J \in \big[ 0, J_{\max} \big]$ for each transmitting packet.
\begin{table}
\label{Table:Parameters}
\centering \caption{Table of Notations and Parameters}
\renewcommand{\arraystretch}{.8}
% \vspace{-.35cm}
\begin{tabular}{cl}
 \hline \hline
  Notation & Description \\
  \hline 
  $J_{\text{ave}}$,\ $J_{\text{TH}}$			& Jammer's average power, jamming threshold\\
  $\mathcal{J}$, $\mathcal{J}_T$			& Jammer's and transmitter's action sets\\
  $\boldsymbol{J}_{(N_J+1)\times 1}$			& Jamming power vector\\
  $ \boldsymbol{x}_{(N_T+1)\times 1} \in \mathbf{X}$	& Transmitter's mixed-strategy\\
  $ \boldsymbol{y}_{(N_J+1)\times 1} \in \mathbf{Y}$	& Jammer's mixed-strategy\\
  $\mathbf{X,Y}$ 					& Mixed-strategy sets\\
  $Z$, $\boldsymbol{Z}_{(N_T+1)\times 1}$		& Transmitter's payoff matrix and payoff vector\\
  $ Z(\boldsymbol{x},\boldsymbol{y})$, $ Z(J_{\text{ave}})$	& Expected payoffs\\
  \hline \hline
\end{tabular}
\end{table}
% 
%
% 
% ******************************************************************************************************************
% 
\section{Game Model}
\label{Sec:Game_Model}
In this section we formulate the jamming problem introduced in section \ref{Sec:System_Model} in a game-theoretic context. We introduce the players, their respective strategies and define a general utility function and payoff matrix that could be used as a measure of performance in a wide rang of jamming problems.
At the physical layer, the interaction between the legitimate users (the transmitter-receiver pair) an the adversary (the jammer) is often modeled as a zero-sum game in order to capture their conflicting goals \cite{Manshaei2013}. 
We use the two-player zero-sum game framework to model the problem with the additional constraint that the jammer must maintain an overall average jamming power. We show how this additional constraint affects jammer's mixed-strategy set and makes the game model different from a typical zero-sum game. 
%
%***************************************************************************************************
%
\subsection{The Mixed-strategy Sets} 
\label{SubSec:Mixed_Strategy_Set}
The jammer has the option to select its operating power in any given packet from the set of discrete values of available jamming powers, arbitrarily placed between $0$ and $J_{\max}$. We assume $(N_J + 1)$ distinct power levels, or pure strategies, are available to the jammer (the size of the jammer's action sets). In the most general case, the jammer's action set is a set of different jamming power levels in the interval $\big[0,J_{\max} \big]$. We denote this set by $\mathcal{J}$
\begin{equation*}% \label{Eq:Jammer_Pure_Set}
 \mathcal{J} = \Big\{0 \leq J_j \leq J_{\max},\ 0 \leq j \leq N_J\ \text{and}\ J_j \neq J_k \ \text{for}\ j\neq k\Big\}
\end{equation*}
Without loss of generality, we assume the power levels are sorted in an increasing order and $\big\{0, J_{\max} \big\} \subset \mathcal{J} $, i.e.
\begin{equation*}% \label{Eq:Power_Levels}
  J_0 = 0 <  \dots < J_j < \dots < J_{N_J}=J_{\max} \qquad 0<j<N_J 
\end{equation*}
For simplicity, we place the possible jamming power levels in a vector and form the jammer's pure strategy
column vector, $\boldsymbol{J}$, where
\begin{equation}\label{Eq:Jamming_Power_Vector}
  \boldsymbol{J}^T   = \begin{bmatrix} J_0 & \cdots & J_j & \cdots & J_{N_J} \end{bmatrix}_{1\times(N_J+1)}
\end{equation}
and $^T$ indicates transposition. Unlike typical zero-sum games in which there are no other constraints on the mixed-strategies, in our model, the jammer's mixed-strategy must satisfy the additional average power constraint $J_{\text{ave}}$, where $J_{\text{ave}} \leq J_{\max}$. Hence, in our constrained game model, not all mixed-strategies (and not even those pure strategies that are greater than $J_{\text{ave}}$) are feasible strategies \cite{Owen}. If we let $\boldsymbol{y}$ denote the jammer's mixed-strategy vector% 
\footnote{The probability distribution vector on the set of jammer's pure strategies.}
and let $\mathbf{Y}$ be the standard $\scriptstyle (N_J + 1)$-simplex, for a typical zero-sum game, we have the following relations
\begin{equation}
  \boldsymbol{y}^T 
 = \left[ \begin{smallmatrix} y_0 & \dots & y_j & \dots & y_{N_J} \end{smallmatrix} \right] 
  \in \mathbf{Y}: \ \
  \sum_{j=0}^{N_J} y_j = 1 \ \
  \begin{aligned} 
  \scriptstyle
   & y_j \geq 0 \\
   & 0\leq j \leq N_J \\
  \end{aligned}
\end{equation}
By using the jammer's pure strategy vector we define the new mixed-strategy set, $\mathbf{Y}_{\text{LE}|J_{\text{ave}}}$, for our constrained game as 
\begin{equation}\label{Eq:Jammer_Mixed_Set_LE}
 \mathbf{Y}_{\text{LE}|J_{\text{ave}}} = \{ \boldsymbol{y} \in \mathbf{Y} \big| \; \boldsymbol{y}^T \cdot
 \boldsymbol{J} \leq J_{\text{ave}} \}
\end{equation}
where $\mathbf{Y}_{\text{LE}|J_{\text{ave}}}$ is a subset of the $\scriptstyle (N_J+1)$-simplex which includes all mixed-strategies with an average power less than or equal to $J_{\text{ave}}$.

Since by introducing the new constrained mixed-strategy set, defined in equation \eqref{Eq:Jammer_Mixed_Set_LE}, we are eliminating some mixed-strategies that could have been otherwise selected, we must first establish the existence of the Nash equilibrium for this game model. This game is not a typical zero-sum game with a finite number of pure strategies for which the existence of the NE is guaranteed. In section~\ref{SubSec:Nash_Existence}, we prove the existence of the NE for the constrained game, where the jammer's mixed-strategy set is limited to  $\mathbf{Y}_{\text{LE}|J_{\text{ave}}}$.

The transmitter's pure strategy set, or equivalently his action set, is a set of $(N_T+1)$ discrete transmission parameters (e.g., power, rate, etc.). We assume each strategy from transmitter's action set can withstand up to a certain level of jamming power, we indicate this jamming power level by $J_T$ to distinguish it from the jammer's actual jamming power, $J$. We assume the packet that is being transmitted with this strategy can be fully recovered at the receiver for any jamming power less than or equal to $J_T$ but will be completely lost for jamming powers greater than $J_T$. 
This assumption is inspired by Shannon's channel capacity theorem which states that reliable communication at a given rate is possible when the noise power is below a certain level and becomes impossible if the noise power exceeds that value. Since corresponding to each transmitter pure strategy there exists a certain jammer power below which reliable transmission is possible, we can define a one-to-one relation between transmitter pure strategies and corresponding jammer power levels $J_T$ (also, see Section \ref{SubSec:Payoff} for more explanation about this assumption). We choose to use these jammer power values as representatives of transmitter's pure strategies.
As a result, the transmitter's pure strategy set can be defined as
\begin{equation}\label{Eq:Transmitter_Pure_Set}
 \mathcal{J}_T = \Big\{0 \leq J_{T,i} \leq J_{\max}, 0 \leq i \leq N_T \Big\}
\end{equation}
WLOG, we assume $J_{T,0} = 0$ and $J_{T,N_T} = J_{\max}$, i.e., the transmitter's highest and lowest payoffs correspond to the jammer's lowest $(J=0)$ and highest $(J=J_{\max})$ jamming powers. 
The transmitter's uses his available transmission parameters (or the equivalent jamming powers from the set $\mathcal{J}_T$) according to a probability distribution (his mixed-strategy) and his goal is to find an optimal strategy to maximize the expected performance of the communication link.
We use column vector $\boldsymbol{x}$, to indicate the transmitter's mixed-strategy vector.
\begin{equation}\label{Eq:Trans_Mixed_Set}
 \boldsymbol{x}^T = \begin{bmatrix} x_0 & \dots & x_i & \dots & x_{N_T}
 \end{bmatrix}_{1\times(N_T+1)} \in \mathbf{X}
\end{equation}
where $\mathbf{X}$ is the standard $\scriptstyle (N_T+1)$-simplex.
%
%***************************************************************************************************
%
\subsection{The Utility Function and The Payoff Matrix}
\label{SubSec:Payoff}
Because transmissions occur in the presence of an adversary, recovery of the transmitted
information/packets at the receiver is not always guaranteed. Since each strategy from the transmitter's action set can sustain up to a certain level of jamming power, denoted by $J_T$, packets can be recovered only when
the actual jamming power, $J$, is less than or equal to $J_T$,
i.e., if and only if $J \leq J_T$. Therefore, the utility function of the game (the payoff to the transmitter), $Z(J_T,J)$, can be modeled as
\begin{equation}
\label{Eq:Utility_Function}
 Z \big(J_T,J \big) =
    \begin{cases}
      Z(J_T)  & J_T \geq J\\
      0       & J_T<J\\
    \end{cases} \quad J_T \in \mathcal{J}_T,\ J \in \mathcal{J}
\end{equation}
where $Z(J_T)$ represents the payoff of the communication link under jamming. The function $Z(J_T)$ assigns a positive value to each strategy from the transmitter's action set and is intuitively a strictly decreasing function of $J_T$, i.e., the payoff to the transmitter decreases when the jamming power increases, i.e.
\begin{equation}\label{Eq:Decreasing_Z_I}
 Z_0 > \dots > Z_i > \dots > Z_{N_T}; \qquad  0<i<N_T
\end{equation}
where $Z_i \stackrel{\Delta}{=} Z(J_T = J_i)$ and $J_i \in \mathcal{J}_T $. 
Even though $Z(J_T)$ could be any arbitrary decreasing function of $J_T$, defining $Z(J_T)$ based on the channel capacity is a common practice in the games involving a transmitter-receiver pair and an adversary \cite{Giannakis08,Koorosh2012,Altman07,Manshaei2013}.

Given that our game model is a constrained zero-sum two-player game, the payoff to the jammer is the negative of the
transmitter's payoff. Furthermore, we can formulate the payoffs (for each pure strategy pair) in a payoff matrix where the transmitter and the
jammer would be the row and the column players respectively. The resulting payoff matrix, $Z$, is in general a non-square matrix and from equation \eqref{Eq:Utility_Function}, we see that the non-zero elements of each row of $Z$ are equal. We will show in Lemma \ref{Lem:N_T} that WLOG, we can assume that $Z$ is a square, lower triangular matrix with equal non-zero entries in each rows, i.e.,
\begin{equation}\label{Eq:Game_Matrix} \renewcommand{\arraystretch}{.2}
  Z_{(N_T+1)\times(N_T+1)} =  
  \left[
  \begin{smallmatrix}
    Z_0     & 0       		& \cdot \cdot &         & 	& 		& 0      \\
    :       &         		& &         &         		&  	  	& :      \\
    Z_i     & \cdot \cdot  	& & Z_i     & 0       		& \cdot \cdot	& 0      \\
    :       &         		& &         &         		&  		& :      \\
    Z_{N_T} & \cdot \cdot 	& & Z_{N_T} & \cdot \cdot	& 		& Z_{N_T}\\ 
  \end{smallmatrix} 
  \right]
\end{equation}
and the expected payoff of the game for the mixed-strategy pair $(\boldsymbol{x},\boldsymbol{y})$ is
\begin{equation}\label{Eq:Expected_Payoff}
  Z(\boldsymbol{x},\boldsymbol{y}) = \boldsymbol{x}^T Z \ \boldsymbol{y}, \qquad
  \begin{aligned}
    & \boldsymbol{y} \in \mathbf{Y}_{\text{LE}|J_{\text{ave}}}\\ 
    & \boldsymbol{x} \in \mathbf{X}\\
  \end{aligned}
\end{equation}
%****************************************************************************
\begin{lemma}\label{Lem:N_T}
  Let $Z$ be the payoff matrix in the constrained two-player zero-sum game defined by the utility function
  \eqref{Eq:Utility_Function}. The payoff matrix obtained by removing the dominated strategies is a square lower
  triangular matrix with size less than or equal to $\min \big\{ N_T, N_J \big\}$.
\end{lemma}
\begin{proof}
See \cite{Koorosh2012} Section 3.3.
\end{proof}

As a consequence of  Lemma \ref{Lem:N_T}, we need to consider only square matrices,
which simplifies further development. In the next section we will study the
outcome of the game under jammer's different values of average power. WLOG, in the following sections, we assume the size of $Z$ matrix is $N_T+1$.
%
% 
% ******************************************************************************************************************
% 
\section{Game Characterization}
\label{Sec:Game_Characterization}
In this Section, we study two basic properties of the game. We will show that although we have put an additional constraint on the jammer's mixed-strategy set, the existence of Nash Equilibrium is still guaranteed.

Additionally, we show that by randomizing his strategy, the jammer can force the transmitter to operate at the lowest payoff, given that the average jamming power exceeds a certain threshold, $J_{\text{TH}}<J_{\max}$. We will also derive an upper bound for this jamming power threshold.
%
%***************************************************************************************************
%
\subsection{Existence of the Nash Equilibrium}
\label{SubSec:Nash_Existence}
For every zero-sum game with finite set of pure strategies, there exists at least one  (pure or mixed) Nash equilibrium (NE) such that no player can do better by unilaterally deviating from his strategy \cite{Owen}. 
In our game model, we are assuming an additional constraint on the jammer's mixed-strategy set; the jammer must maintain a maximum average jamming power. 
This additional assumption changes the jammer's mixed-strategy set from a standard $n$-simplex to a subset of it.
Therefore, the existence of the NE for this non-typical zero-sum game must be established.  In the following lemma, we show that the existence of the NE under the additional constraint is guaranteed.
 \begin{lemma}
  \label{Lemma:Nash_Existence}
   For the constrained two-player zero-sum game defined by the utility function $Z \left( J_T,J \right)$, given in \eqref{Eq:Utility_Function}, and the payoff matrix $Z$, given by \eqref{Eq:Game_Matrix}, there exists at least one NE in the form of transmitter's mixed-strategy, $ \boldsymbol{x} \in \mathbf{X} $, and the jammer's mixed-strategy, $\boldsymbol{y} \in \mathbf{Y}_{\mathrm{LE}|J_{\mathrm{ave}}}$.
 \end{lemma}
%***************************************
\begin{proof}
  See Appendix \ref{Appendix_2}.
\end{proof}
%
%
%***************************************************************************************************
%
\subsection{Existence of Jamming Power Threshold}
\label{SubSec:Existence_of_Jamming_Threshold}
The following theorem proves the existence of a jamming power threshold and gives an upper bound for it. In Section \ref{Sec:Game_Analysis} we use the Theorem \ref{Theorem:Threshold} to derive the optimal mixed-strategies for the transmitter and the jammer in the general case.
\begin{theorem}
  \label{Theorem:Threshold}
  Let us assume we have a constrained two-player zero-sum game defined by the utility function $Z(J_T,J)$, given in \eqref{Eq:Utility_Function}, the payoff matrix, $Z$, given in \eqref{Eq:Game_Matrix}, the transmitter's mixed-strategy $ \boldsymbol{x} \in \mathbf{X} $, and the jammer's mixed-strategy
    $\boldsymbol{y} \in \mathbf{Y}_{\mathrm{LE}|J_{\mathrm{ave}}}$, given in \eqref{Eq:Jammer_Mixed_Set_LE}. Then, there exists a jammer threshold power $J_{\mathrm{TH}}$, $0 < J_{\mathrm{TH}} < J_{\max}$ 
  such that, if $J_{\mathrm{ave}} \geq J_{\mathrm{TH}}$ then,  there exists $ \boldsymbol{y}^* \in \mathbf{Y}_{\mathrm{LE}|J_{\mathrm{ave}}}$ and
  \begin{equation}\label{Eq:Trans_Optimal_Mixed_For_Powerful_Jammer}
    \begin{aligned}
      & \boldsymbol{x}^{*^T} = \begin{bmatrix} \boldsymbol{0}_{1\times N_T} &  1 \\ \end{bmatrix}_{1\times(N_T+1)}\\
      & Z(\boldsymbol{x}^*,\boldsymbol{y}^*) = Z_{N_T}
    \end{aligned}
  \end{equation}
\end{theorem}
Where $\boldsymbol{x}^*,\boldsymbol{y}^*$ are the transmitter's and jammer's optimal mixed-strategies, respectively, and $ Z(\boldsymbol{x}^*,\boldsymbol{y}^*) $ is the game value at the NE (we use these notations throughout the paper).
%   \footnote{We will denote the transmitter and jammer optimal strategies by $\boldsymbol{x}^*$ and $\boldsymbol{y}^*$, respectively, and the game value at the NE by $ Z(\boldsymbol{x}^*,\boldsymbol{y}^*) $.}.

Theorem \ref{Theorem:Threshold} states that there exits a jamming power threshold, $J_{\text{TH}}$, such that if the jammer's average power exceeds~$J_{\text{TH}}$ then the transmitter's optimal mixed-strategy is to use the strategy corresponding to jammer's maximum power, resulting in lowest payoff; as if the jammer was using its maximum jamming power all the time. 
\begin{proof}
 Assume the jammer is using a mixed-strategy, $\dot{\boldsymbol{y}}$, which is not necessarily optimal, and is defined by
\begin{equation}
\label{Eq:Non_Optimal_Strategy}
  \begin{aligned}
    \dot{\boldsymbol{y}}^{T} 
			 & =  \begin{bmatrix}
			           \dot{y}_0 & 0 & \cdots & 0 & \dot{y}_{N_T} \\
			          \end{bmatrix}_{1 \times (N_T+1)} \\
			 & = Z_{N_T} \left[ \begin{smallmatrix} 
				  {\scriptstyle Z^{-1}_{0}} & {\scriptstyle 0} & \cdots & 0 & \left( Z^{-1}_{N_T} - Z^{-1}_{0} \right)  
				 \end{smallmatrix} \right]
  \end{aligned}
\end{equation}
It can easily be verified that \eqref{Eq:Non_Optimal_Strategy} is indeed a valid probability distribution (since by assumption we have $ {\textstyle Z_{N_T} < Z_0 }$). Let $\dot{J}_{\text{ave}}$ be the average jamming power of the strategy $\dot{\boldsymbol{y}}$;
\begin{equation}
\label{Eq:Average_Jamming_Power_For_Non_Optimal_Strategy}
 \dot{J}_{\text{ave}} = \sum^{N_T}_{j=0} \dot{y}_j J_j = \left( 1 - \frac{Z_{N_T}}{Z_0} \right) J_{\max} < J_{\max}
\end{equation}
Furthermore, assume the transmitter is using an arbitrary mixed-strategy $\boldsymbol{x} \in \mathbf{X}$ against jammer's strategy $\dot{\boldsymbol{y}}$ defined in \eqref{Eq:Non_Optimal_Strategy}. Define $Z(\boldsymbol{x}, \dot{\boldsymbol{y}})$ to be the expected payoff of the game for the mixed-strategy pair $\left( \boldsymbol{x}, \dot{\boldsymbol{y}} \right)$; 
\begin{align}
\label{Eq:Expected_Payoff_of_the_Game_for_Non_Optimal_Startegy}
    Z(\boldsymbol{x}, \dot{\boldsymbol{y}}) & = \boldsymbol{x}^T Z \ \dot{\boldsymbol{y}} %\notag \\
		 = Z_{N_T} \left[ (x_0 + x_{N_T}) + \sum^{N_T-1}_{i=1} \left(\frac{Z_i}{Z_0} \right) x_i \right] \notag \\
		& \qquad \leq Z_{N_T} \left(x_0 + x_{N_T}\right) \leq Z_{N_T}
\end{align}
Since by assumption we have $\frac{Z_i}{Z_0} < 1$ for all $0 < i \leq N_T$ and $ 0 \leq x_i \leq 1$ for all $0\leq i \leq N_T$. Additionally, the equality in \eqref{Eq:Expected_Payoff_of_the_Game_for_Non_Optimal_Startegy} holds if and only if $x_0+x_{N_T} =1$
Thus, by using the mixed-strategy, $\dot{\boldsymbol{y}}$, and an average power $\dot{J}_{\text{ave}} \leq J_{\max}$ given in \eqref{Eq:Average_Jamming_Power_For_Non_Optimal_Strategy}, the jammer can force a payoff at most equal to the transmitter's lowest payoff, $Z_{N_T}$. 
\end{proof}
The jamming power given in \eqref{Eq:Average_Jamming_Power_For_Non_Optimal_Strategy} is not necessarily the lowest possible threshold. In section \ref{SubSec:Optimal_Mixed_Strategies_J_Ave_Greater_Than_J_TH} we provide a closed form expression for the lowest average jamming power, $J_\text{TH}$, (jamming \emph{threshold}) that can force the payoff $Z_{N_T}$.
% Consequently, a rational jammer with an average jamming power $J_{\text{ave}} \geq J_{\text{TH}}$, only needs an average  power equal to $J_{\text{TH}}$ to force the transmitter to operate at his lowest payoff,~$Z_{N_T}$.
%
% 
% 
% 
%
% 
% ******************************************************************************************************************
% 
\section{Game Analysis}
\label{Sec:Game_Analysis}
In this section we study the optimal strategies for the transmitter and the jammer. We divide this section into two subsections; the case where the average jamming power is less than the jamming threshold as defined in Section \ref{Sec:Game_Characterization}, and the case where it is greater than or equal to the jamming threshold.
As we will show in Appendix \ref{Appendix_2}, the standard linear programing techniques that are used to solve standard two-player zero-sum games can be appropriately modified to solve two-player zero-sum games with linear constraints. 
Therefore, even though all the results in this section are derived analytically, all constrained two-payer zero-sum games can also be solved numerically. 
% 
%
%***************************************************************************************************
%
\subsection{Optimal Mixed-Strategies for $J_{\mathrm{ave}} < J_{\mathrm{TH}}$}
\label{SubSec:Optimal_Mixed_Strategies_J_Ave_Less_Than_J_TH}
We start by developing the optimal strategies for specific values of average jamming power, denoted by $J_{\text{ave},m}$, $m =  0,\cdots , N_T-1$. As we will show, for these specific average jamming powers, the expected payoff of the game at the NE is equal to $Z_m$, $m = 0,\cdots , N_T-1$.
Later we extend our analytic results to the case where average jamming power is not necessarily equal to $J_{\text{ave},m}$.
% general case and provide a closed form expression for the expected value of the game at NE as a function of the jammer's average power.
%
%***************************************************************************************************
%
\subsubsection{$J_{\mathrm{ave}} = J_{\mathrm{ave},m} <J_{\mathrm{TH}}$, for some $m = 0,\cdots , N_T-1$}\label{SubSubSec:J_Ave_m}
Since the average jamming power is less than the jamming threshold, the expected value of the game at the NE is in the interval $( Z_{N_T},Z_0]$. Assume, for now, the average jamming power, $J_{\text{ave}} < J_{\text{TH}}$, is such that the optimal strategy for the jammer is to use $(m+1)$ of his pure strategies (the support set of the jammer's mixed-strategy is $y_0 \leq y \leq y_m$), i.e.,
\begin{equation} 
\renewcommand{\arraystretch}{.2}
\setlength{\arraycolsep}{2pt}
 \begin{aligned}
  & \boldsymbol{y}^T = \begin{bmatrix}
                        y_0 & \dots & y_j & \dots & y_m & \boldsymbol{0}_{1\times (N_T-m)}
                       \end{bmatrix}_{1\times (N_T+1)}\\
 \end{aligned}
\end{equation}
where $0\leq m<N_T$ and $\boldsymbol{0}_{1 \times (N_T-m)}$ indicates a row vector of $(N_T-m)$ zeros. 
Intuitively, as the jammer's average power increases, he is able to use pure strategies with higher jamming power.
Define the mixed-strategy of the jammer, $\widehat{\boldsymbol{y}}$, as 
\begin{equation} 
  \widehat{\boldsymbol{y}}^T = Z_m
  \left[
 \begin{smallmatrix}\label{Eq:Y_Hat}
  {\scriptstyle Z_0^{-1}}&
  {\scriptstyle \cdots} &
  {\scriptstyle Z_j^{-1} - Z_{j-1}^{-1}} &
  {\scriptstyle \cdots} &
  {\scriptstyle Z_m^{-1} - Z_{m-1}^{-1}} &
  {\textstyle \boldsymbol{0}_{1\times (N_T-m)}}
 \end{smallmatrix}%_{(N_T+1) \times 1}
 \right]
\end{equation}
It can be easily verified that \eqref{Eq:Y_Hat} is indeed a probability distribution. The average power corresponding to $\widehat{\boldsymbol{y}}$ is
\begin{equation}\label{Eq:Y_Hat_Averag_Power}
 J_{\text{ave},m} \stackrel{\Delta}{=} \boldsymbol{J}^T \widehat{\boldsymbol{y}} = Z_m \sum^{m}_{j=1} \left( Z_j^{-1} - Z_{j-1}^{-1} \right) J_j  
\end{equation}
Assume the jammer is using the mixed-strategy $\widehat{\boldsymbol{y}}$ 
% given in equation \eqref{Eq:Y_Hat} 
against the transmitter's arbitrary distribution, $\boldsymbol{x} \in \mathbf{X}$. The resulting expected payoff of the game is
\begin{equation}
\label{Eq:Expected_Payoff_for_Y_Hat}
% \small
  \begin{aligned}
   Z \big( \boldsymbol{x}, \widehat{\boldsymbol{y}} \big) = \boldsymbol{x}^TZ\boldsymbol{y}
					         = Z_m \cdot \boldsymbol{x}^T \begin{bmatrix}
					                                       {\textstyle \boldsymbol{1}_{(m+1)\times 1}}\\
					                                       {\scriptstyle Z_{m+1}Z_m^{-1} }\\
					                                       \vdots\\
									       {\scriptstyle Z_j Z_m^{-1} }\\
									       \vdots\\
					                                       {\scriptstyle Z_{N_T}Z_m^{-1} }\\
					                                      \end{bmatrix}  \leq Z_m
  \end{aligned} 
\end{equation}
Since $Z_i < Z_j$ for $i > j$, it is clear from \eqref{Eq:Expected_Payoff_for_Y_Hat} that the optimal strategy for the transmitter, against $\widehat{\boldsymbol{y}}$ , is to use at most the same number of pure strategies, $(m+1)$, i.e.
\begin{equation}
\label{Eq:X_M}
\setlength{\arraycolsep}{2pt}
% \small
    \boldsymbol{x}^T = \begin{bmatrix}
                        x_0 & \dots & x_i & \dots & x_m & \boldsymbol{0}_{1\times (N_T-m)}
                       \end{bmatrix}_{1\times (N_T+1)}
\end{equation}
any strategy other than \eqref{Eq:X_M} results in a lower expected payoff for the transmitter.

Since $ \widehat{\boldsymbol{y}}$ given in \eqref{Eq:Y_Hat} is not necessarily an optimal mixed-strategy for the jammer, the optimal strategy would result in an expected payoff less than $Z_m$. Therefore, $Z_m$ can be used as an upper bound for the game value at NE and all mixed-strategies with average jamming power $J_{\text{ave}} = J_{\text{ave},m}$. We present this fact in the following lemma.
\begin{lemma}\label{Lemma:Game_Upper_Bound}
Assume the jammer's average power is given by \eqref{Eq:Y_Hat_Averag_Power}. Then for the constrained two-player zero-sum game defined in theorem \ref{Theorem:Threshold}, the following inequality holds
\begin{equation}\label{Eq:Game_Upper_Bound}
  Z \big( \boldsymbol{x}^*, \boldsymbol{y}^* \big) \leq Z \big( \boldsymbol{x}, \widehat{\boldsymbol{y}} \big) \leq Z_m \qquad
      \begin{aligned}
       & J_{\mathrm{ave}} = J_{\mathrm{ave},m}\\
       & \forall \boldsymbol{x} \in \mathbf{X}\\
      \end{aligned}
\end{equation}
\end{lemma}

As discussed above, the optimal strategy for the transmitter against $\widehat{\boldsymbol{y}}$ is to use, at most, $(m+1)$ of his strategies. Define the mixed-strategy $\widehat{\boldsymbol{x}}$ for the transmitter as
\begin{equation}
\label{Eq:X_Hat}
\setlength{\arraycolsep}{2pt}
% \small
  \begin{aligned}
    \widehat{\boldsymbol{x}}^T & = \begin{bmatrix}
				    x_0 & \cdots & x_i & \cdots & x_m & \boldsymbol{0}_{1\times (N_T-m)}
			           \end{bmatrix}\\
			       & = \begin{bmatrix}
				    b_0Z_0^{-1} & \cdots & b_iZ_i^{-1} & \cdots & b_mZ_m^{-1} & \boldsymbol{0}_{1\times (N_T-m)}
			           \end{bmatrix}\\
  \end{aligned}
\end{equation}
where selection of $b_i$'s will be discussed later. 
Since $\widehat{\boldsymbol{x}}$ is a probability distribution, we have
\begin{equation} \label{Eq:Sum_X_Hat}
 \sum^{m}_{i=0} x_i = \sum^{m}_{i=0} b_i Z_i^{-1} = 1 \qquad b_i > 0 \quad i=0,\dots,m
\end{equation}

Assume the transmitter is using the mixed-strategy $\widehat{\boldsymbol{x}}$, which is not necessarily optimal, against jammer's arbitrary strategy $\boldsymbol{y} \in \mathbf{Y}_{\text{LE}|J_{\text{ave}}}$. The expected payoff of the game would be
\begin{equation}
\label{Eq:Expected_Payoff_for_X_Hat}
 \begin{aligned}
  & Z\big( \widehat{\boldsymbol{x}}, \boldsymbol{y} \big) = \widehat{\boldsymbol{x}}^T Z \ \boldsymbol{y} 
	 = y_0  \sum^{m}_{i=0} b_i + \dots + y_j \sum^{m}_{i=j} b_i + \dots y_m b_m \\
	 & = B \sum^{m}_{j=0} y_j -  \left( y_0b_0 + \dots + y_j \sum^{j-1}_{i=0} b_i + \dots + y_m \sum^{m-1}_{i=0} b_i \right)
 \end{aligned}
\end{equation}
where $B = \sum^{m}_{i=0} b_i$. Since we assumed $J_{\text{ave}} < J_{\text{TH}}$, a rational jammer would use all his available power, i.e.
\begin{equation*}
 \boldsymbol{J}^T \boldsymbol{y} = \sum^{m}_{i=1} J_j y_j = J_{\text{ave}}\ \text{where} \
  \begin{aligned}
   &\boldsymbol{y} \in \mathbf{Y}_{\text{LE}|J_{\text{ave}}} \\
   &J_1 < \dots < J_j < \dots < J_m \\
  \end{aligned}
\end{equation*}
Let the sum of the terms in the parentheses in relation \eqref{Eq:Expected_Payoff_for_X_Hat} be proportional to the jammer's average power, i.e.,
\begin{equation}\label{Eq:B_I_Solution}
 \begin{aligned}
  & \sum^{j-1}_{i=0} b_i = d\cdot J_j \quad i=0,\dots,m-1 \qquad \text{where} \ d>0\\
  & \qquad \Rightarrow b_i = d(J_{i+1} - J_i) \quad i=0, \dots, m-1
 \end{aligned}
\end{equation}
then \eqref{Eq:Expected_Payoff_for_X_Hat} becomes independent of jammer's strategy and the expected payoff of the game is
\begin{equation}\label{Eq:Z_of_X_Hat_Y}
 \begin{aligned}
  Z\big( \widehat{\boldsymbol{x}}, \boldsymbol{y} \big) = \widehat{\boldsymbol{x}}^T Z \boldsymbol{y} 
	= & \sum^{m}_{i=0} b_i - d J_{\text{ave}}\\
	= & b_m + d \left( J_m - J_{\text{ave}} \right)\\
 \end{aligned}
 \quad \forall \boldsymbol{y} \in \mathbf{Y}_{\text{LE}|J_{\text{ave}}}
\end{equation}
It is clear from \eqref{Eq:Z_of_X_Hat_Y} that a rational jammer should use all his available power, $J_{\text{ave}}$, to achieve the lowest possible expected payoff against $\widehat{\boldsymbol{x}}$.
If we substitute $b_i$'s from~\eqref{Eq:B_I_Solution} in~\eqref{Eq:Sum_X_Hat} we~have
{\small
\begin{equation*}
 \begin{aligned}
  b_mZ_m^{-1} + d \sum^{m-1}_{j=0} \left(J_{j+1} - J_j \right) Z_j^{-1} =& 1\\
  b_m + dZ_m \left( \sum^{m-1}_{j=0} J_{j+1}Z_j^{-1} - \sum^{m-1}_{j=0}J_jZ_j^{-1} \right) =& Z_m\\
  b_m + dZ_m \left( \sum^{m}_{j=1} J_{j}Z_{j-1}^{-1} - \sum^{m}_{j=1}J_jZ_j^{-1} +J_mZ_m^{-1} \right) =& Z_m\\
\end{aligned}
\end{equation*}%
\begin{equation*}
 \begin{aligned}
%   b_m + dZ_m \left( \sum^{m}_{j=1} J_{j}Z_{j-1}^{-1} - \sum^{m}_{j=1}J_jZ_j^{-1} +J_mZ_m^{-1} \right) =& Z_m\\
  b_m + d \left( J_m - \sum^{m}_{j=1} \left(Z_j^{-1} - Z_{j-1}^{-1} \right)J_j \right) =& Z_m\\
 \end{aligned}
\end{equation*}
}
and finally from \eqref{Eq:Y_Hat_Averag_Power} 
\begin{equation}\label{Eq:Z_M_Relation}
 b_m + d(J_m - J_{\text{ave},m}) = Z_m
\end{equation}
If we substitute $J_{\text{ave}}$  with $J_{\text{ave},m}$ in \eqref{Eq:Z_of_X_Hat_Y} and use \eqref{Eq:Z_M_Relation} we have
\begin{equation}\label{Eq:Game_Lower_Bound_1}
  Z\big( \widehat{\boldsymbol{x}}, \boldsymbol{y} \big) = b_m + d(J_m - J_{\text{ave},m}) = Z_m \ \ \ \forall \boldsymbol{y} \in \mathbf{Y}_{\text{LE}|J_{\text{ave},m}}
\end{equation}
Since $\widehat{\boldsymbol{x}}$ given in \eqref{Eq:X_Hat} is not necessarily an optimal mixed-strategy, using the optimal strategy for the transmitter would result in an expected payoff greater than $Z_m$. Therefore, $Z_m$ could be used as a lower bound for the game value at the NE and all mixed-strategies with average jamming power $J_{\text{ave}} = J_{\text{ave},m}$. We present this result in the following lemma.

\begin{lemma}\label{Lemma:Game_lower_Bound}
Assume the jammer's average power is given by \eqref{Eq:Y_Hat_Averag_Power}. Then for the constrained two-player zero-sum game defined in theorem \ref{Theorem:Threshold}, the following inequality holds
\begin{equation}\label{Eq:Game_Lower_Bound_2}
  Z \big( \boldsymbol{x}^*, \boldsymbol{y}^* \big) \geq Z \big( \widehat{\boldsymbol{x}}, \boldsymbol{y} \big) \geq Z_m \qquad
      \begin{aligned}
       & J_{\mathrm{ave}} = J_{\mathrm{ave},m}\\
       & \forall \boldsymbol{y} \in \mathbf{Y}_{\mathrm{LE}|J_{\mathrm{ave},m}}\\
      \end{aligned}
\end{equation}
\end{lemma}

However, from lemma \ref{Lemma:Game_Upper_Bound} and for $J_{\text{ave}} = J_{\text{ave},m} $ we know the game value cannot be more than $Z_m$ and hence the game value at the NE is indeed $Z_m$ and $\widehat{\boldsymbol{x}}$ and $\widehat{\boldsymbol{y}}$ given in \eqref{Eq:X_Hat} and \eqref{Eq:Y_Hat} are optimal mixed-strategies for the transmitter and the jammer, respectively. 
% \footnote{Although \eqref{Eq:X_Hat} and \eqref{Eq:Y_Hat} are the optimal mixed-strategies but they are not unique.}.
Since by assumption, $m<N_T$, we have $J_{m+1} \in \mathcal{J}$, therefore we can let 
\begin{equation}
 b_m = d \left( J_{m+1} - J_m \right)
\end{equation}
and from \eqref{Eq:Game_Lower_Bound_1} we have
\begin{equation}\label{Eq:D_And_B_i_Relations}
 \begin{aligned}
  & d = \frac{Z-m}{J_{m+1} - J_{\text{ave},m}} \\
  & b_i = \frac{J_{i+1} - J_i}{J_{m+1} - J_{\text{ave},m}}Z_m\\
 \end{aligned}
 \quad \text{for} \quad
 \begin{aligned}
  & 0 \leq m < N_T\\
  & i = 0, \dots, m\\
 \end{aligned}
\end{equation}
substituting \eqref{Eq:D_And_B_i_Relations} in \eqref{Eq:X_Hat}, the transmitter's optimal mixed-strategy becomes
\begin{equation}
\renewcommand{\arraystretch}{.2} 
\label{Eq:Trans_Optimal Strategy}
 \widehat{\boldsymbol{x}} = \begin{bmatrix}
                             x_0\\
                             \vdots\\
                             x_i\\
                             \vdots\\
                             x_m\\
                             \boldsymbol{0}_{N_T-m)\times 1}\\
                            \end{bmatrix}
                             = {\textstyle \frac{Z_m}{J_{m+1} - J_{\text{ave},m}}} \begin{bmatrix} 
								{\scriptstyle \left( J_1 - J_0 \right) Z_0^{-1}}\\
								\vdots\\
								{\scriptstyle \left( J_{i+1} - J_i \right) Z_i^{-1}}\\
								\vdots\\
								{\scriptstyle \left( J_{m+1} - J_m \right) Z_m^{-1}}\\
				\boldsymbol{0}_{N_T-m)\times 1}\\
                            \end{bmatrix}
\end{equation}

We summarize the results  derived so far in the following theorem.
\begin{theorem}\label{Theorem:Weak_Jammer_Theorem}
 Consider the constrained two-player zero-sum game defined by utility function \eqref{Eq:Utility_Function}, payoff matrix \eqref{Eq:Game_Matrix}, and transmitter and jammer mixed-strategy sets $\mathbf{X}$ and $\mathbf{Y}_{\mathrm{LE}|J_{\mathrm{ave},m}}$, defined in \eqref{Eq:Trans_Mixed_Set} and \eqref{Eq:Jammer_Mixed_Set_LE}, respectively. Then, the expected value of the game at the Nash-Equilibrium is
\begin{equation*}
 Z \big( \boldsymbol{x}^*, \boldsymbol{y}^* \big) = Z_m
\end{equation*}
and $J_{\mathrm{ave},m}$ is given by
 \begin{equation*}
  J_{\mathrm{ave},m} = Z_m \sum^{m}_{j=1} \left( Z_j^{-1} - Z_{j-1}^{-1} \right) J_j \qquad 0 \leq m < N_T
 \end{equation*}
Furthermore, $\boldsymbol{x}^*$ and $\boldsymbol{y}^*$ are given by \eqref{Eq:Trans_Optimal Strategy} and \eqref{Eq:Y_Hat}, respectively. 
\end{theorem}

If we define the effective jamming power, $J_{\text{eff}}$, to be the jamming power a \emph{reactive non-strategic jammer} (i.e., a jammer that uses only pure strategies) would need to force the same operating point at the NE ($Z_m$ in this case) then, the effective jamming power becomes
\begin{equation}\label{Eq:Effective_Jamming_Power_Week_Jammer}
 J_{\text{eff}} = J_m > J_{\text{ave},m} \qquad \text{for}\ 0 < m < N_T
\end{equation}
which means that randomizing helps the jammer to achieve the same performance as the reactive non-strategic jammer with less average jamming power.
%
%***************************************************************************************************
%
\subsubsection{Optimal Strategies for the General Case}\label{SubSubSec:J_Ave}
Now that we have established the optimal mixed-strategies for $J_{\text{ave}} = J_{\text{ave},m}$, we consider the more general case where the jammer's average power is not necessarily equal to $J_{\text{ave},m}$ for some $0\leq m < N_T$. Obviously we have $J_{\text{ave},m} \leq J_{\text{ave}} < J_{\text{ave},m+1} $ for some $ m \in \big\{ 0,\dots, N_T-1 \big\} $. Let
\begin{equation}
\label{Eq:J_Ave_General}
 \begin{aligned}
  J_{\text{ave}} & = J_{\text{ave},m} + \varepsilon= Z_m \sum^{m}_{j=1} \left( Z_j^{-1} - Z_{j-1}^{-1} \right) J_j + \varepsilon \\
        & \qquad \qquad \qquad \text{for some}\  m \in \big\{ 0,\dots, N_T-1 \big\}\\
 \end{aligned}
\end{equation}
Define the mixed-strategy $\widehat{y}$ for the jammer as 
\begin{equation}
\renewcommand{\arraystretch}{.2}
\label{Eq:Y_Hat2}
  \widehat{\boldsymbol{y}} =a \cdot Z_m
 \begin{bmatrix}
  Z_0^{-1} \\
  \vdots \\
  Z_j^{-1} - Z_{j-1}^{-1} \\
  \vdots \\
  Z_m^{-1} - Z_{m-1}^{-1} \\
  \widehat{y}_{m+1}\\
  \boldsymbol{0}_{(N_T-m-1)\times 1}
 \end{bmatrix}_{(N_T+1) \times 1}
\end{equation}
Since $\widehat{\boldsymbol{y}}$ is a probability distribution we have
\begin{equation}
  \sum^{m+1}_{j=0} y_j = aZ_m \left( Z_m^{-1} +\widehat{y}_{m+1} \right)= 1 \Rightarrow a = \frac{1}{1+Z_m\widehat{y}_{m+1}}
\end{equation}
and from \eqref{Eq:Y_Hat2} we can rewrite the expression for the jammer's average power as
\begin{equation}
 \begin{aligned}
  J_{\text{ave}} = & \sum^{m+1}_{j=1} y_jJ_j = aJ_{\text{ave},m} + aZ_m\widehat{y}_{m+1}J_{m+1}\\
  \text{from \eqref{Eq:J_Ave_General}} \Rightarrow & \quad aJ_{\text{ave},m} + aZ_m\widehat{y}_{m+1}J_{m+1} = J_{\text{ave},m} + \varepsilon\\
 \end{aligned}
\end{equation}
and hence $a$ and $\widehat{y}_{m+1}$ become
\begin{equation}
%  \begin{aligned}
  \widehat{y}_{m+1} = \frac{J_{\text{ave}}-J_{\text{ave},m}}{J_{m+1} - J_{\text{ave}}} Z_m^{-1} \ \text{and} \
  a  = \frac{J_{m+1} - J_{\text{ave}}}{J_{m+1} - J_{\text{ave},m}}\\
%  \end{aligned}
\end{equation}
Assume the jammer is using the mixed-strategy given in \eqref{Eq:Y_Hat2} against the transmitter's arbitrary strategy. Then the expected payoff of the game is
\begin{equation}
%  \begin{aligned}
 Z\big( \boldsymbol{x}, \widehat{\boldsymbol{y}} \big)  = aZ_m  \boldsymbol{x}^T \begin{bmatrix}
											\boldsymbol{1}_{(m+1)\times 1}\\
											{\scriptstyle Z_{m+1}Z_m^{-1} \widehat{y}_{m+1} }\\
											\vdots\\
											\end{bmatrix}\\
       ={\scriptstyle \frac{J_{m+1} - J_{\text{ave}}}{J_{m+1} - J_{\text{ave},m}}} Z_m \\
%  \end{aligned}
\end{equation}
Where we have used the fact that a rational transmitter would only use up to $(m+1)$ of his strategies since otherwise the expected payoff of the game would be even less. As before, since $\widehat{\boldsymbol{y}}$ is not necessarily an optimal mixed-strategy, the corresponding expected payoff can be used as an upper bound for the game and hence, we have the following lemma.
\begin{table*}\label{Table:Game_Summary}
\centering \caption{Summary of the Results}
\renewcommand{\arraystretch}{.3}
% \vspace{-.35cm}
\begin{tabular}{|c|c|}
 \hline
%   \multicolumn{2}{|c|}{ }\\
$Z\big( \boldsymbol{x}^*, \boldsymbol{y}^* \big) = {\boldsymbol{x}^*}^T Z \boldsymbol{y}^* = \frac{J_{m+1} - J_{\text{ave}}}{J_{m+1} - J_{\text{ave},m}} Z_m 
    \quad 
    J_{\text{ave},m} \leq J_{\text{ave}} < J_{\text{ave},m+1}$ &
    $J_{\text{ave},m} = Z_m \sum^{m}_{j=1} \left( Z_j^{-1} - Z_{j-1}^{-1} \right) J_j 
    \quad 
    m = 0,\cdots, N_T-1$\\
  \hline
   & \\
  Transmitter's Optimal Mixed-Strategy & Jammer's Optimal Mixed-Strategy\\
%   \hline
%   \renewcommand{\arraystretch}{.5}
  $\boldsymbol{x}^* = \begin{bmatrix}
                       x_0\\
		       \vdots\\
		       x_i\\
		       \vdots\\
		       x_m\\
		       \boldsymbol{0}\\
                      \end{bmatrix}_{(N_T+1)\times 1}
		    = \frac{Z_m}{J_{m+1} - J_{\text{ave},m}} \begin{bmatrix}
							(J_1 - J_0) Z_0^{-1}\\
							:\\
							(J_{i+1} - J_i) Z_i^{-1}\\
							:\\
							(J_{m+1} - J_m) Z_m^{-1}\\
							\boldsymbol{0}\\
						       \end{bmatrix}
  $ &
  $
  \boldsymbol{y}^* = \begin{bmatrix}
                       y_0\\
		       \vdots\\
		       y_j\\
		       \vdots\\
		       y_m\\
		       y_{m+1}\\
		       \boldsymbol{0}\\
                      \end{bmatrix}_{(N_T+1)\times 1}
		   = \frac{J_{m+1} - J_{\text{ave}}}{J_{m+1} - J_{\text{ave},m}} Z_m \begin{bmatrix}
									Z_0^{-1 \qquad \qquad  }\\
									:\\
									Z_j^{-1} - Z_{j-1}^{-1}\\
									:\\
									Z_m^{-1} - Z_{m-1}^{-1}\\
									\frac{J_{\text{ave}} - J_{\text{ave},m}}{J_{m+1} - J_{\text{ave}}} Z_m^{-1}\\
									\boldsymbol{0}\\
								       \end{bmatrix}
  $\\
  \hline
\end{tabular}
\end{table*}
\begin{lemma}\label{Lemma:Game_Upper_Bound2}
Let us assume that the jammer's average power, $J_{\mathrm{ave}}$, satisfies $J_{\mathrm{ave},m} < J_{\mathrm{ave}} < J_{\mathrm{ave},m+1} $ for some $ m \in \left\{ 0,\dots, N_T-1 \right\} $. Then for the constrained two-player zero-sum game defined in theorem \ref{Theorem:Threshold}, and for all $\boldsymbol{x} \in \mathbf{X} $ the following inequality holds
\begin{equation}\label{Eq:Game_Upper_Bound2}
  Z \big( \boldsymbol{x}^*, \boldsymbol{y}^* \big) \leq Z \big( \boldsymbol{x}, \widehat{\boldsymbol{y}} \big) \leq \frac{J_{m+1} - J_{\mathrm{ave}}}{J_{m+1} - J_{\mathrm{ave},m}} Z_m \\
\end{equation}
where $\widehat{\boldsymbol{y}}$ is given in \eqref{Eq:Y_Hat2}.
\end{lemma}
Now assume the transmitter is using the same mixed-strategy given in \eqref{Eq:X_Hat}. From \eqref{Eq:Z_of_X_Hat_Y} and \eqref{Eq:D_And_B_i_Relations} we have
\begin{equation}
 \begin{aligned}
  Z\big( \widehat{\boldsymbol{x}}, \boldsymbol{y} \big) & = b_m + d \left( J_m - J_{\text{ave}} \right)
							   = d \left( J_{m+1} - J_{\text{ave}} \right)\\
							  & = \frac{J_{m+1}-J_{\text{ave}}}{J_{m+1}-J_{\text{ave},m}}Z_m\\
 \end{aligned}
\end{equation}
which is the same expression for the upper bound of the game derived in Lemma \ref{Lemma:Game_Upper_Bound2}. As a result we have the following theorem.
\begin{theorem}\label{Theorem:Weak_Jammer_Theorem2}
 Consider a constrained two-player zero-sum game defined by utility function \eqref{Eq:Utility_Function}, payoff matrix \eqref{Eq:Game_Matrix}, and transmitter and jammer mixed-strategy sets $\mathbf{X}$ and $\mathbf{Y}_{\mathrm{LE}|J_{\mathrm{ave}}}$ given in \eqref{Eq:Trans_Mixed_Set} and \eqref{Eq:Jammer_Mixed_Set_LE}, respectively. Assume the jammer's average power $J_{\mathrm{ave}}$ satisfies
 \begin{equation}
  J_{\mathrm{ave},m} < J_{\mathrm{ave}} < J_{\mathrm{ave},m+1} \ \text{for some} \ m \in \left\{ 0,\dots, N_T-1 \right\}
 \end{equation}
Then, the expected value of the game at the NE is
\begin{equation}
 Z \big( \boldsymbol{x}^*, \boldsymbol{y}^* \big) = \frac{J_{m+1}-J_{\mathrm{ave}}}{J_{m+1}-J_{\mathrm{ave},m}}Z_m
\end{equation}
where $\boldsymbol{x}^*$ and $\boldsymbol{y}^*$ are the transmitter's and the jammer's optimal mixed-strategies, respectively, and $J_{\mathrm{ave},m}$ is given by \eqref{Eq:Y_Hat_Averag_Power}.  
Furthermore, $\boldsymbol{x}^*$ and $\boldsymbol{y}^*$ are given by \eqref{Eq:Trans_Optimal Strategy} and \eqref{Eq:Y_Hat2}, respectively.
\end{theorem}
%
%***************************************************************************************************
%
% 
\subsection{Optimal Mixed-Strategies for $J_{\mathrm{ave}} \geq J_{\mathrm{TH}}$}
\label{SubSec:Optimal_Mixed_Strategies_J_Ave_Greater_Than_J_TH}
Let us assume in lemma \ref{Lemma:Game_Upper_Bound} we let $m=N_T$, then we have 
\begin{equation}\label{Eq:Game_Upper_Bound_for_Powerful_Jammer}
  Z \big( \boldsymbol{x}^*, \boldsymbol{y}^* \big) \leq Z_{N_T} \qquad
      \begin{aligned}
       & J_{\text{ave}} = J_{\text{ave},N_T}\\
       & \forall \boldsymbol{x} \in \mathbf{X}\\
      \end{aligned}
\end{equation}
Since the game value cannot be less than $Z_{N_T}$, we conclude that $Z \big( \boldsymbol{x}^*, \boldsymbol{y}^* \big) = Z_{N_T}$. But from theorem \ref{Theorem:Threshold} we know that for $J_{\text{ave}} \geq J_{\text{TH}}$ the game value at NE is also $Z_{N_T}$ and since $J_{\text{ave},N_T}$ is the smallest average jamming power for which the game value at NE is equal to $Z_{N_T}$, we conclude that $J_{\text{TH}} = J_{\text{ave},N_T}$. We summarize this result in the following theorem.
\begin{theorem}\label{Theorem:Powerful_Jammer}
 Consider the constrained two-player zero-sum game defined by utility function \eqref{Eq:Utility_Function}, payoff matrix \eqref{Eq:Game_Matrix}, and transmitter and jammer mixed-strategy sets $\mathbf{X}$, and $\mathbf{Y}_{\mathrm{LE}}$, defined in \eqref{Eq:Trans_Mixed_Set} and \eqref{Eq:Jammer_Mixed_Set_LE}, respectively. Then, there exists a jamming power threshold, $J_{\mathrm{TH}} < J_{\max}$, such that %the expected game value at Nash Equilibrium is
 \begin{equation*}
  Z \big( \boldsymbol{x}^*, \boldsymbol{y}^* \big) = Z_{N_T} \qquad \forall J_{\mathrm{ave}} \geq J_{\mathrm{TH}}
 \end{equation*}
where the value of $J_{\mathrm{TH}}$ is given by
\begin{equation}\label{Eq:Jamming_Threshold}
 J_{\mathrm{TH}} = Z_{N_T} \sum^{N_T}_{j=1} \left( Z_j^{-1} - Z_{j-1}^{-1} \right) J_j  \qquad J_j \in \mathcal{J}
\end{equation}
Furthermore, the jammer's optimal mixed-strategy with the lowest average power, that can achieve the NE is given by
\begin{equation*}
\setlength{\arraycolsep}{2pt}
  \boldsymbol{y}^{*^T} = Z_{N_T}
 \begin{bmatrix}
  {\scriptstyle Z_0^{-1}} &
  {\scriptstyle \cdots} &
  {\scriptstyle Z_j^{-1} - Z_{j-1}^{-1}} &
  {\scriptstyle \cdots} &
  {\scriptstyle Z_{N_T}^{-1} - Z_{N_T-1}^{-1}} &
 \end{bmatrix}_{1\times (N_T+1)}
\end{equation*}
\end{theorem}
In other words, if the average jamming power exceeds $J_{\text{TH}}$ given by \eqref{Eq:Jamming_Threshold}, then the \emph{optimal strategic jammer} (i.e., the jammer which uses optimal mixed-strategy) can force the expected payoff equal to the transmitter's lowest payoff at the NE.
% the transmitter to operate at the strategy corresponding to his lowest payoff
% which corresponds to the case where jammer uses $J_{\max}$ all the time (Barrage noise jamming), i.e.
This expected payoff is equal to the reactive non-strategic jammer with average power $J_{\max}$, i.e.,
% \begin{equation}\label{Eq:Effective_Jamming_Power_Powerful_Jammer}
 $ J_{\text{eff}} = J_{\max} > J_{\text{TH}}$.
% \end{equation}
%
% 
% 
% 
%
% 
% ******************************************************************************************************************
% 
\section{Special Case;  AWGN Channel Capacity as the\\ Utility Function}
\label{Sec:Special_Case}
% \input{Sec_Special_Case2.tex}
% 
% 
%***** OLD TEXT *****
% 
% In this section we study the special case of optimal jamming strategy in an actual wireless communication link under jamming. We use the results of the preceding sections to investigate the optimal strategies for the transmitter and the jammer under different assumptions and constrains. In what follows, we use the AWGN channel capacity as the utility function, i.e., we assume transmission at channel capacity.
% 
%***** NEW TEXT *****
% 
In this section we study two typical jamming scenarios and we show that the framework defined in previous sections can be used to determine the optimal transmission and jamming strategies. Even though our analytical results are not limited to the AWGN channel, for simplicity, we use the packetized AWGN channel model and
we assume packets are long enough that channel capacity theorem could be applied to each packet being transmitted.
% in what follows.
%
% 
%
%***************************************************************************************************
%
\subsection{Special Case I: A Transmitter with Fixed  Rates}%
\label{SubSec:Special_Case_I}%
In this section we study a special case of the game defined in the previous sections. We assume the communication link between the transmitter and the receiver is a sing-hop, packetized (discrete-time), AWGN channel with fixed and known noise variance, $N$. The communication link is being disrupted by an additive adversary. We assume the jammer is an additive Gaussian jammer with flat power spectral density. It can be shown \cite{CoverJ06} that in the AWGN channel with a fixed and known noise variance, an iid Gaussian jammer is the most effective jammer in minimizing the capacity between the transmitter and the receiver. The effect of the Gaussian jammer on the communication link is reduction of the effective signal to noise ratio (SNR) at the receiver from $\frac{P_T}{N}$ to $\frac{P_T}{N+J}$, where $J$ represents the jammer power (variance) and $P_T$ is the transmitter power.

The transmitter has a rate adaptation block which allows him to transmit at $(N_T+1)$ \emph{different} but \emph{fixed} rates according to the system's design specifications. Any rate other than these given rates is not a feasible option for the transmitter. Additionally, the transmission rates are bounded between a minimum and a maximum transmission rate denoted by $R_{\min}$ and $R_{\max}$, respectively, i.e.,  $R_{\min} \leq R_i \leq R_{\max}$ for $i=0, \dots, N_T$.
Without loss of generality, we assume the rates are sorted in a decreasing order. Hence, the transmitter's action set becomes
\begin{equation}\label{Eq:Trans_Action_Set_Arbitrary_Rates}
 \mathcal{R} = \big\{ {\scriptstyle R_0=  R_{\max} > \dots > R_i > \dots > R_{N_T} = R_{\min}} \big\}
\end{equation} 
The transmitter's goal is to maximize the achievable expected transmission rate over the channel. Assuming that transmission at channel capacity is possible, and from the capacity of discrete-time AWGN channel, the transmission power must at least be equal to
\begin{equation}\label{Eq:Trans_Power}
 \begin{aligned}
  & R_{\max}  = R_0 = \frac{1}{2} \log \left(1 + \frac{P_T}{N} \right) \ (\text{nats/transmission})\\
  & \Rightarrow \qquad P_T = N \left( e^{2R_0} -1 \right) \\
 \end{aligned}
\end{equation}
Throughout the rest of this section, we assume transmission at channel capacity and the transmission power is fixed and given by \eqref{Eq:Trans_Power}. 
Given that the channel noise variance is assumed to be fixed and known, corresponding to each transmission rate $R_j \in \mathcal{R} $ there exists a certain jammer power, $\widehat{J_j} \geq 0$, below which reliable transmission is possible, i.e.
\begin{equation}\label{Eq:Jamming_Powers_Corresponding_Rates}
 \begin{aligned}
  & R_j = \frac{1}{2} \log \left( 1 + \frac{P_T}{N+ \widehat{J_j}} \right); \ \ j = 0,\dots, N_T \\
%   & \Rightarrow \qquad \widehat{J_j} = \frac{P_T}{e^{2R_j} - 1} - N\\
  & \Rightarrow \qquad \widehat{J_j} = N\frac{e^{2R_0} - e^{2R_j}}{e^{2R_j} - 1 };  \ j = 0,\dots, N_T \\
 \end{aligned}
\end{equation}
With this notation, we can define a one to one correspondence between transmitter rates and jammer power levels. Therefore, we can use $\widehat{J}_j$ given in \eqref{Eq:Jamming_Powers_Corresponding_Rates} and/or $R_j$ for $j=0,\dots, N_T$ to refer to transmitter strategies interchangeably. 

Assume the jammer's goal is to force the transmitter to operate at his lowest rate, $R_{N_T}$, while keeping the lowest possible average and maximum jamming power. As a result of Lemma \ref{Lem:N_T}, the jammer does not need to use more strategies than the transmitter i.e., he only needs $(N_T+1)$ jamming power levels. Consider the following action set for the jammer
\begin{equation}
% \label{Eq:Jammer_Action_Set}
\label{Eq:Jamming_Powers_Corresponding_Rates_Plus_Delta}
 \mathcal{J} = \Big\{ J_0, J_1, \dots, J_j, \dots, J_{N_T} \Big\}
\end{equation}
% where $J_j$ is given by
\begin{equation*}%\label{Eq:Jamming_Powers_Corresponding_Rates_Plus_Delta}
  J_j = 
  \begin{cases}
    0 & j = 0\\
    \widehat{J}_{j-1} + \delta N= N\left( \delta + \frac{e^{2R_0} - e^{2R_{j-1}}}{e^{2R_{j-1}} - 1 } \right)  & j = 1,\dots, N_T\\
  \end{cases}
\end{equation*}
The term $ \delta N $ with $\delta > 0 $ is an extra added jamming power to the non-zero jamming powers to make sure that $R_{j-1}$ is greater than channel capacity for the jamming power $J_j$. Since the transmitter's goal is to achieve the maximum possible expected transmission rate and the jammer's goal is to minimize the same expected value, we can define the utility function based on the capacity of the discrete-time AWGN channel, i.e.
\begin{equation}\label{Eq:Utility_Function_Rate}
 C \big( R_j , J_j \big) = 
  \begin{cases}
   R_j & \widehat{J}_j \geq J_j\\
   0   & \widehat{J}_j < J_j\\
  \end{cases}
  \qquad R_j \in \mathcal{R} \ \text{and} \ J_j \in \mathcal{J}
\end{equation}
The utility function defined in \eqref{Eq:Utility_Function_Rate} has the same format of the general utility function defined in \eqref{Eq:Utility_Function} and it can be easily verified that the $R_j$ as defined in \eqref{Eq:Jamming_Powers_Corresponding_Rates} is a strictly decreasing function of $j$, or equivalently $\widehat{J}_j$. Hence, we can directly apply Theorem \ref{Theorem:Powerful_Jammer} to find the minimum average jamming power or the jamming power threshold, $J_{\text{TH}}$.
Substituting \eqref{Eq:Jamming_Powers_Corresponding_Rates} and \eqref{Eq:Jamming_Powers_Corresponding_Rates_Plus_Delta} in \eqref{Eq:Jamming_Threshold} and simplifying the result, the jamming threshold becomes
\begin{equation}\label{Eq:Jamming_Threshold_For_Fixed_Rates}
 \begin{aligned}
  J_{\text{TH}} & = R_{N_T} \sum^{N_T}_{j=1} \left(\frac{1}{R_j} - \frac{1}{R_{j-1}} \right)J_j
	  = N\delta \left( 1 - \frac{R_{N_T}}{R_0} \right)  \\
	 & \quad + N R_{N_T} \sum^{N_T}_{j=1} \left[ \left(\frac{1}{R_j} - \frac{1}{R_{j-1}} \right) \frac{e^{2R_0} - e^{2R_{j-1}}}{e^{2R_{j-1}} - 1 } \right]\\
 \end{aligned}
\end{equation}
and the jammer's optimal mixed-strategy with the minimum average jamming power that achieves the NE is given by
\begin{equation}
\renewcommand{\arraystretch}{.2}
\setlength{\arraycolsep}{2pt}
  \boldsymbol{y}^{*^T} = R_{N_T}
 \begin{bmatrix}
  {\scriptstyle R_{0}^{-1}} &
  {\scriptstyle \cdots} &
  {\scriptstyle R_j^{-1} - R_{j-1}^{-1}} &
  {\scriptstyle \cdots} &
  {\scriptstyle R_{N_T}^{-1} - R_{N_T-1}^{-1}} &
 \end{bmatrix}_{1\times (N_T+1)}
\end{equation}
We define the jammer's \emph{randomization gain} to be the power advantage that he gains for switching from pure-strategies (i.e., reactive non-strategic jammer) to optimal mixed-strategies. With this definition, the randomization gain becomes the ratio of the jammer's pure strategy power ($J_{\max}$ in this case) over the average power of the optimal mixed-strategy that forces the same expected payoff at the NE, i.e.
\begin{equation}
 \text{Randomization Gain} = \frac{J_{\max}}{J_{\text{TH}}} = \frac{J_{N_T}}{J_{\text{TH}}} > 1
\end{equation}

To provide a numerical example, assume, a single-hop jamming resilient communication system that uses the typical rates of the IEEE 802.11x standard i.e., the available coded data rates of the communication system are
% for practical reasons, the jammer cannot use system weaknesses to his advantage and his only option is to jam the system at the physical layer (i.e., the system is jamming resilient or the jammer is not aware of the communication scheme details). 
% To provide more realistic results, we use typical rate of the IEEE 802.11a standard. Assume the coded data rates of the system are given by
\begin{equation}
 \mathcal{R}_{a} = \Big\{ 54,\ 51,\ 48,\ \cdots ,\ 12,\ 9,\ 6 \Big\} \qquad (\text{Mb/s})
\end{equation}
Figure \ref{Fig:Fixed_Rates_AWGN_Transmitter}(top) shows the randomization gain of the optimal jammer as a function of the expected transmission rate at the NE for this example. The figure is sketched for continuous AWGN channel and typical values of $N$ and $\delta$. For this typical example, the randomization gain of the jammer is
% 
% if we let $N=1$ and $\delta = 0.01$ and plug the rates in \eqref{Eq:Jamming_Threshold_For_Fixed_Rates} and use the continuous AWGN channel capacity 
% (instead of the discrete AWGN channel capacity), then
% the randomization gain of the optimal jamming scheme to force the transmitter to operate at 48 Mb/s becomes
\begin{equation}
 3\ \text{dB} \approx 1.8 \leq \text{Randomization Gain}  \leq 16 \approx 12.5 \ \text{dB}
\end{equation}
Figure \ref{Fig:Fixed_Rates_AWGN_Transmitter} (top) also provides a comparison between the reactive non-strategic jammer and the optimal strategic jammer. As expected, the optimal strategic jammer requires less average power than the reactive non-strategic jammer to force the same expected rate at the NE. Figure \ref{Fig:Fixed_Rates_AWGN_Transmitter} (bottom) shows a typical optimal transmission strategy for the transmitter in this case.
%
% 
% 
% **** OLD TEXT *****
% To show the effectiveness of this jamming scheme in a practical case, consider the IEEE 802.11a wireless standard. The coded data rates of IEEE 802.11a standard are
% 
%*********************
\begin{figure}
 \centering
 \includegraphics[width = 3in]{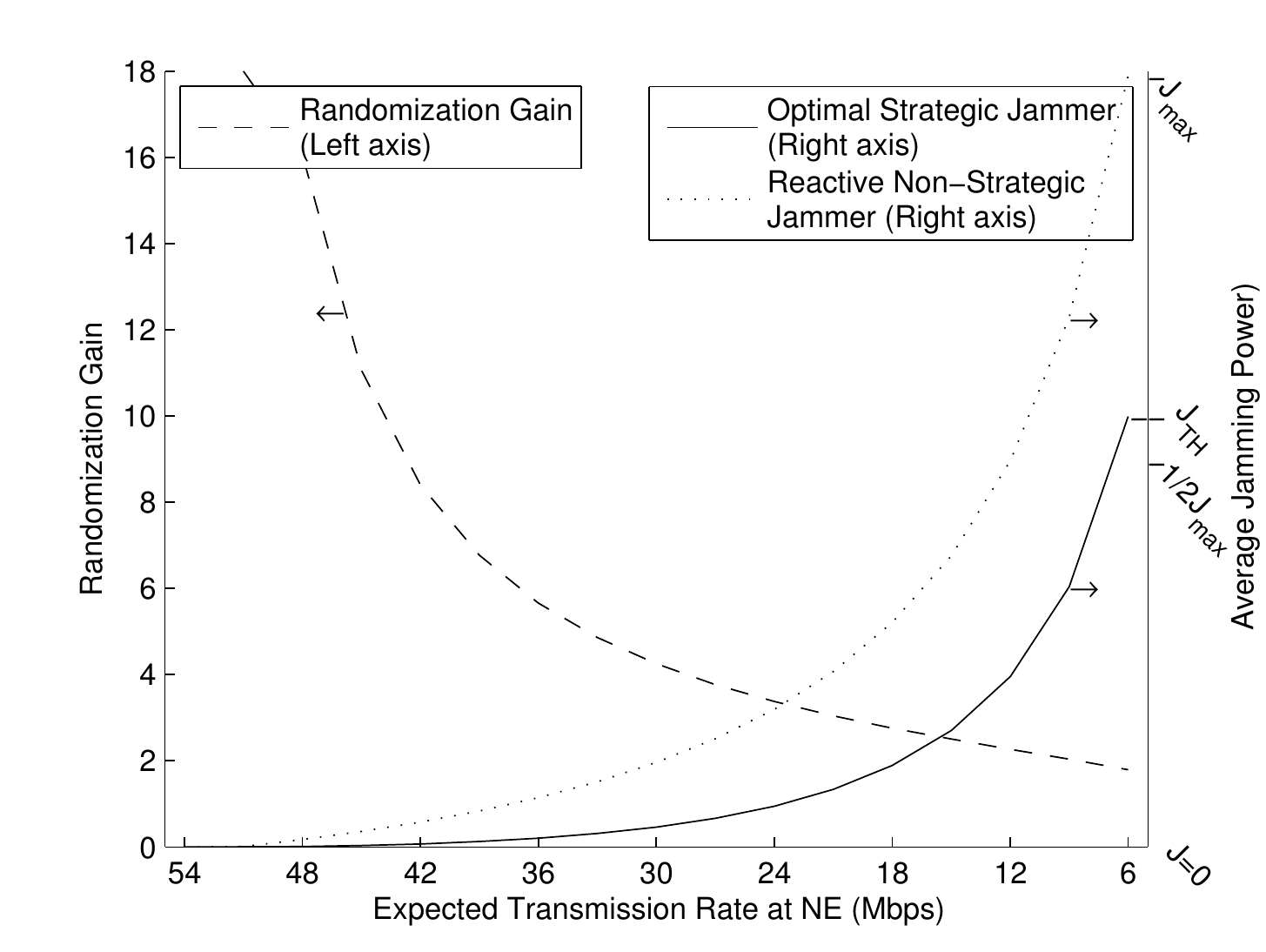}
 \includegraphics[width = 2.0in]{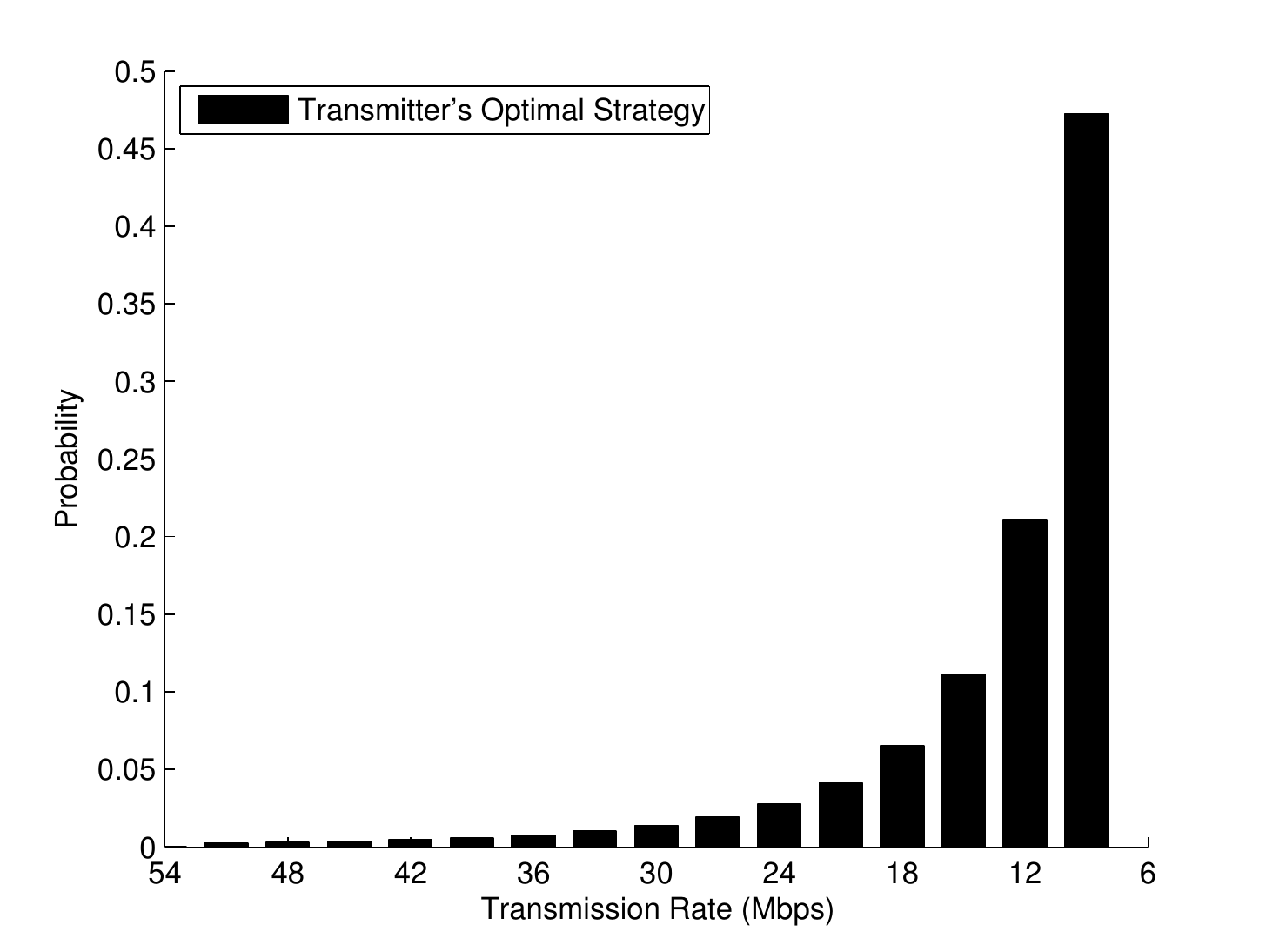}
 \caption{Jammer's randomization gain and average power as a function of expected rate at the NE (top), transmitter's typical optimal mixed-strategy for $J_{\text{ave}} < J_{\text{TH}}$ (bottom).}
 \label{Fig:Fixed_Rates_AWGN_Transmitter}
\end{figure}
%
%***************************************************************************************************
%
\subsection{Special Case II: Equally Spaced Jamming Powers}
\label{SubSec:Special_Case_II}
Consider a communication link with the same setup defined in section \ref{SubSec:Special_Case_I}. 
Assume the jammer is using $(N_T+1)$ discrete jamming power levels equally spaced in the interval $\left[0,J_{\max}\right]$, i.e., the jammer's action set is
\begin{equation}\label{Eq:Equally_Spaced_Jamming_Powers}
 \mathcal{J} = \Big\{ J_j = \frac{j}{N_T} J_{\max}; \quad 0\leq j \leq N_T \Big\}
\end{equation}

The transmitter has a rate adaptation block which allows him to transmit at any arbitrary rate. Assume the transmitter's goal is to maximize the achievable expected transmission rate over the discrete-time AWGN channel. As a result of lemma \ref{Lem:N_T}, the optimal strategy for the transmitter is to use, at most, $(N_T+1)$ rates where each rate corresponds to one of the jammer pure strategies. 
Assuming that transmission at the AWGN channel capacity is possible, we can define the achievable transmission rate based on the discrete-time AWGN channel capacity when the signal to noise ratio $\frac{P_T}{N}$ is replaced by the signal to noise plus jamming power ratio $\frac{P_T}{N+J_j}$. In this case, the transmitter's action set becomes
\begin{equation}
\label{Eq:Rates_forEqually_Spaced_Jamming_Powers}
 \mathcal{R} = \Big\{ R_i = 
  {\textstyle 
  \frac{1}{2} \log \left( 1 + \frac{P_T}{N + \frac{i}{N_T}J_{\max}} \right) \ 0\leq i\leq N_T }
  \Big\}
\end{equation}
In this special case, the jamming power set representing the transmitter's action set, $\mathcal{J}_T$, is identical to the jammer's action set, $\mathcal{J}_T = \mathcal{J}$. 
Since the transmitter's goal is to achieve the maximum expected transmission rate over the channel, we can define the utility function of the game based on the AWGN channel capacity, i.e.,
\begin{equation}\label{Eq:R_of_J_T}
 Z\big( J_T \big) = R \big( J_T \big) = \frac{1}{2} \log \left( 1 + \frac{P_T}{N + J_T} \right) \quad J_T \in \mathcal{J}
\end{equation}

Given that at rates higher than capacity reliable communication is impossible and since
$R \big( J_T \big)$ defined in \eqref{Eq:R_of_J_T} is a strictly decreasing function, the framework defined in section \ref{Sec:Game_Model} can be applied to this special case. Thus, the results derived in section \ref{Sec:Game_Analysis} can be used to determine the optimal strategies and the expected value of the game at NE.

Assuming the jammer's average power, $J_{\text{ave}}$, satisfies $J_{\text{ave},m}$ $\leq J_{\text{ave}} <$ $J_{\text{ave},m+1}$ for some $m = 0,\cdots, N_T-1$, the optimal mixed-strategies for the transmitter and the jammer simplify to
\begin{equation*}
\setlength{\arraycolsep}{2pt}
%  \begin{aligned}
  \boldsymbol{x}^{*^T} = \left( \sum^{m}_{i=0} R_i^{-1} \right)^{-1}\begin{bmatrix}
                                                                {\scriptstyle R_0^{-1}}   &
								{\scriptstyle \cdots}   &
								{\scriptstyle R_i^{-1}} &
								{\scriptstyle \cdots}   & 
								{\scriptstyle R_m^{-1}} &
								{\textstyle \boldsymbol{0}_{1\times (N_T-m)}} &                                                               \end{bmatrix}%_{(N_T+1) \times 1}\\
\end{equation*}
% and
\begin{equation*}
% \begin{aligned}
 \boldsymbol{y}^*   =  
 {\scriptstyle \left[ (m+1) - N_T \frac{J_{\text{ave}}}{J_{\max}} \right]\left( \sum^{m}_{i=0} R_i^{-1} \right)^{-1}} 
				      \left[
				      \begin{smallmatrix}
				      {\scriptstyle R_0^{-1}} 			\\
				      {\scriptstyle :} 			\\
				      {\scriptstyle R_j^{-1}- R_{j-1}^{-1}} 	\\
				      {\scriptstyle :} 			\\
				      {\scriptstyle R_m^{-1} - R_{m-1}^{-1}} 	\\
				      {\scriptstyle \frac{J_{m+1} - J_{\text{ave}}}{J_{m+1} - J_{\text{ave},m}}R_m^{-1}} \\
				      \boldsymbol{0}_{(N_T-m-1) \times1} \\                                                               \end{smallmatrix}
				      \right]
%  \end{aligned}
\end{equation*}
where $J_{\text{ave},m}$ is given by
\begin{equation}
 J_{\text{ave},m} = \frac{1}{N_T}J_{\max} \left[ (m+1) - R_m \sum^{m}_{i=0} R_i^{-1} \right]
\end{equation}
The expected value of the game at the NE, as function of the jammer's average power, is given by
\begin{equation}
 R\big( J_{\text{ave}} \big) = \left[ (m+1) - N_T \frac{J_{\text{ave}}}{J_{\max}} \right] \left( \sum^{m}_{i=0} R_i^{-1} \right)^{-1}
\end{equation}
%*********************************************

In this special case it can be shown that an upper bound for jamming power threshold is given by
% In this special case, it can be show that, for all $P_T,$ $N,$  $ N_T > 0$, equation \eqref{Eq:J_TH_Max} is maximized for $i=0$ (see Appendix \ref{Appendix_1}). Therefore, the upper bound for the jamming threshold simplifies to
\begin{equation}\label{Eq:J_TH_Upper_Bound}
{\textstyle
 J_{\text{TH,U}} =   \frac{1}{2}\frac{N_T+1}{N_T} \left( 1 - \frac{R_{N_T}}{R_0} \right) J_{\max} < \frac{1}{2} J_{\max}\ \ \ N_T\gg 1
 }
\end{equation}
and a simple strategy and an approximation to the jammer's optimal strategy that achieves this bound is given by
\begin{equation}
{\textstyle
 \dot{\boldsymbol{y}}^{T} =  \begin{bmatrix}
				      \dot{y}_0 & \dot{y} & \cdots & \dot{y} \\
				    \end{bmatrix}
  ; \
   \dot{y}_0 = \frac{R_{N_T}}{R_0}  \ \text{and} \
   \dot{y} = 1- \frac{1}{N_T} \dot{y}_0 
   }
\end{equation}
\begin{proof}
 Assume the jammer is using a mixed-strategy, $\dot{\boldsymbol{y}}$, according to%
 \footnote{We will use the term \emph{semi-uniform} to refer to this class of pmf.}
\begin{equation}
\label{Eq:Semi_Uniform}
\dot{\boldsymbol{y}} = \begin{bmatrix} y_0 \\ y \\ \vdots \\ y \end{bmatrix}_{(N_T + 1) \times 1} \text{where} \ \
  \begin{aligned}
%   & \dot{\boldsymbol{y}}^T = \begin{bmatrix} y_0 & y & \dots & y \end{bmatrix}_{1 \times (N_T + 1)}\\
  & y_0 = 1 - \frac{2N_T}{N_T + 1} \cdot \frac{J_{\text{ave}}}{J_{\max}}\\
  & y   = \frac{2}{N_T + 1} \cdot \frac{J_{\text{ave}}}{J_{\max}}
 \end{aligned}
\end{equation}
It can easily be verified that 
$ \dot{\boldsymbol{y}} \in \mathbf{Y}_{\text{LE}|J_{\text{ave}}} $. Furthermore, assume the transmitter is using an arbitrary mixed-strategy in which the probability associated with the payoff $R_i$, $0\leq i \leq N_T$ is denoted by $x_i$.
Define $R(\boldsymbol{x}, \dot{\boldsymbol{y}})$ to be the expected payoff of the game for the transmitter's arbitrary mixed-strategy, $\boldsymbol{x}$, against jammer's mixed-strategy defined in \eqref{Eq:Semi_Uniform}
\begin{equation}
{\textstyle
 R(\boldsymbol{x}, \dot{\boldsymbol{y}}) = R_{-i,N_T} + R_{N_T}x_{N_T} + R_ix_i \Pr \big[ J\leq J_T = J_i \big]
 }
\end{equation}
where $R_{-i,N_T}$ denotes the partial expected payoff resulting from all pure strategies except for the $i$'th and $N_T$'th strategies. In order to improve his payoff, the transmitter deviates from his current strategy, $\boldsymbol{x}$, to a new strategy, $ \boldsymbol{x}' $, where $x_{N_T}' = x_{N_T} + \delta $ and $ x_i' = x_i - \delta $ and $ 0 < \delta $. Define $R(\boldsymbol{x}', \dot{\boldsymbol{y}})$ to be the expected payoff for the new strategy. 
\begin{equation}
 \begin{aligned}
    R(\boldsymbol{x}', \dot{\boldsymbol{y}}) = & R_{-i,N_T} + R_{N_T} \left( x_{N_T} + \delta \right) \\
					       & + R_i \left( x_i -\delta \right) \Pr \Big[ J\leq J_T = J_i \Big]\\
					     = &  R(\boldsymbol{x}, \dot{\boldsymbol{y}}) + \delta \big[ R_{N_T} - R_i \left( y_0 + iy \right) \big]
 \end{aligned}
\end{equation}
Let $\Delta R $  be the difference in the expected payoff of the game caused by deviating to the new strategy, i.e.,
\begin{equation}\label{Eq:Delta_Z}
 \begin{aligned}
  \Delta R & = R(\boldsymbol{x}', \dot{\boldsymbol{y}}) - R(\boldsymbol{x}, \dot{\boldsymbol{y}}) \\
	   & = \delta \left[ R_{N_T} - R_i \left( 1-2\frac{N_T-i}{N_T+1}\cdot \frac{J_{\text{ave}}}{J_{\max}} \right) \right]\\
 \end{aligned}
\end{equation}
where $ \delta > 0 $ and $ 0\leq i < N_T $. Assume (for now) that $\Delta R \geq 0 $ then we can rewrite \eqref{Eq:Delta_Z} as
\begin{equation}\label{Eq:J_Ave_Inequality}
{\textstyle
 J_{\text{ave}} \geq J_{\max} \left( 1 - \frac{R_{N_T}}{R_i} \right) \left( \frac{1}{2} \frac{N_T + 1}{N_T - i} \right) \stackrel{\Delta}{=} U_i 
%  \ 0\leq i < N_T
 }
\end{equation}
Define $J_{\text{TH,U}}$ as 
\begin{equation}\label{Eq:J_TH_Max}
 J_{\text{TH,U}} = \max_{0\leq i < N_T} J_{\max} 
{\textstyle
 \left( 1 - \frac{R_{N_T}}{R_i} \right) \left( \frac{1}{2} \frac{N_T + 1}{N_T - i} \right)
 }
\end{equation}
then for $J_{\text{ave}} > J_{\text{TH}}$ and for all $\delta>0 $ and $0\leq i < N_T$ the inequality in \eqref{Eq:J_Ave_Inequality} is satisfied and hence $\Delta R > 0$. As a result, the transmitter can improve his payoff by dropping the probability of his $i$'th strategy and adding it to his $N_T$'s strategy (the strategy with the lowest payoff). Since the $i$'th strategy was chosen arbitrary, the transmitter can improve his expected payoff by dropping probability from all strategies, except the $N_T$'th strategy, and adding them to the $N_T$'th strategy. This process can be continued until all probabilities are accumulated in  $x_{N_T}$ and no further improvement to the expected payoff is possible.

In general, it can be shown that \eqref{Eq:J_TH_Max} is maximized for $i=0$ (See Appendix \ref{Appendix_1}) which results in the desired upper bound and the mixed-strategy given in \eqref{Eq:J_TH_Upper_Bound} and \eqref{Eq:Semi_Uniform}, respectively. The penalty in using the semi-uniform strategy instead of the optimal mixed-strategy is that the jammer requires greater jamming power to force the same expected NE.
\end{proof}
\begin{figure}
 \centering
 \includegraphics[width = 3in]{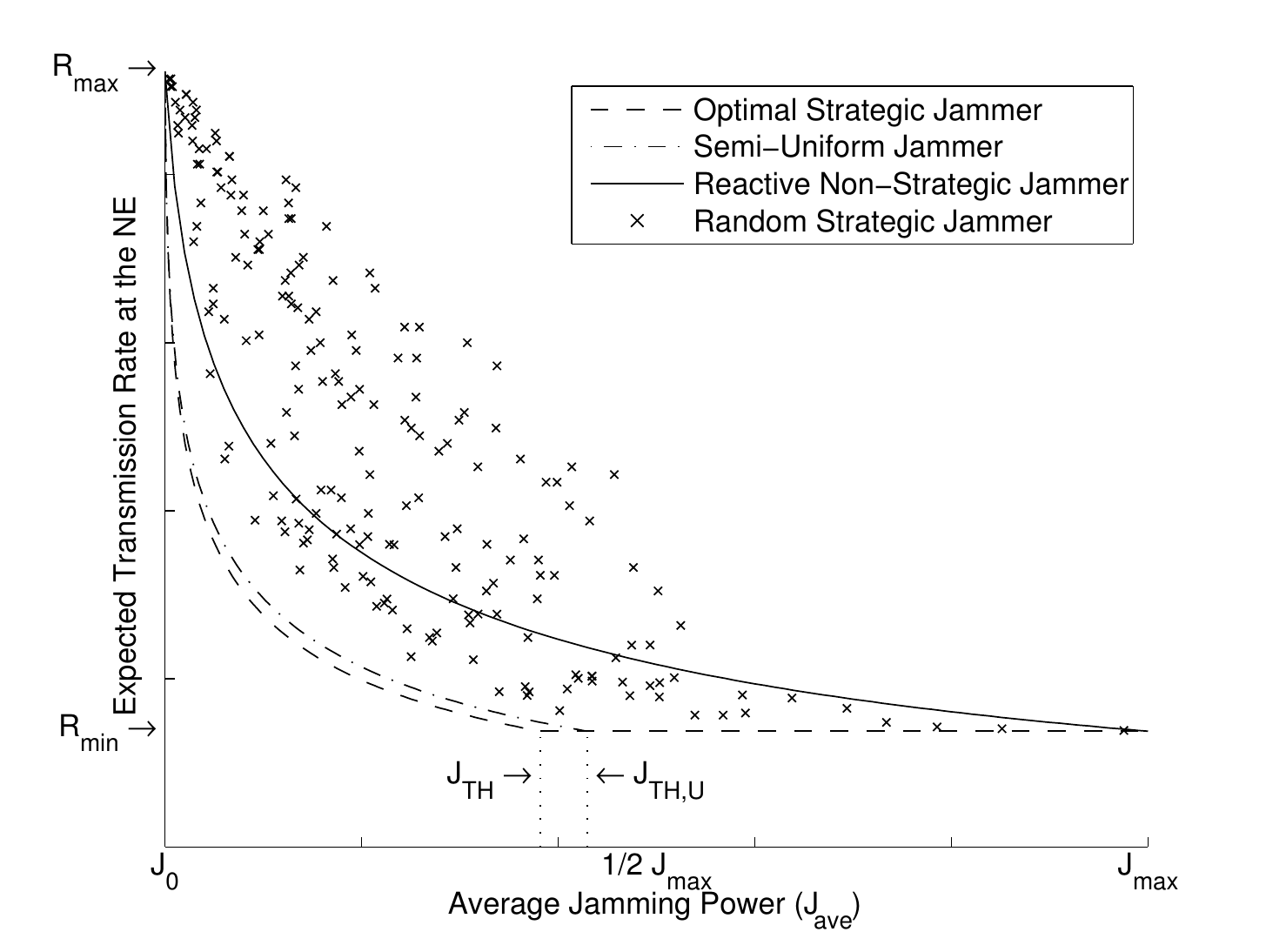}
 \includegraphics[width = 2.0in]{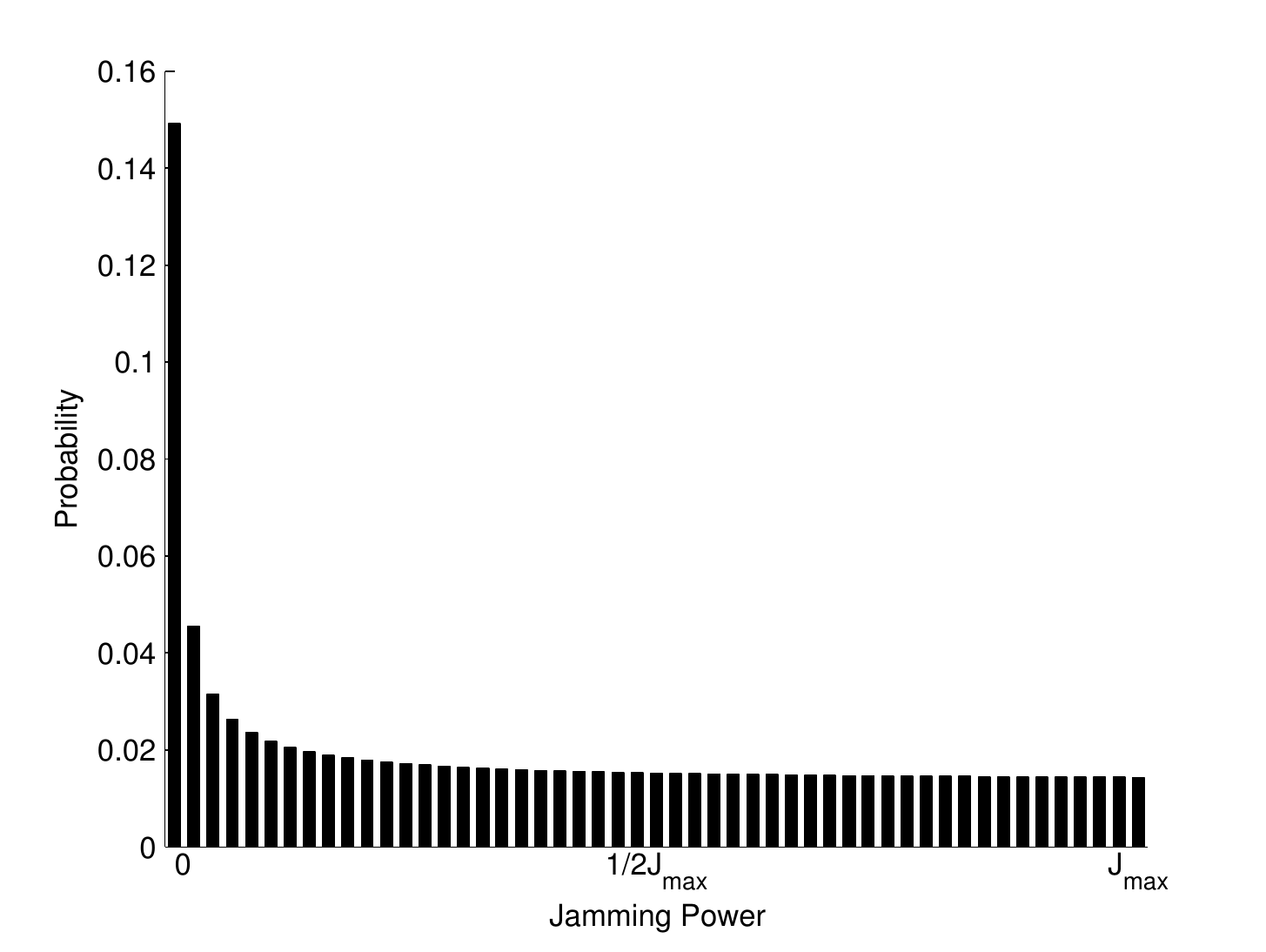}
 \vspace{-.5cm}
 \caption{$R\big( J_{\text{ave}} \big) $ as a function of $J_{\text{ave}}$ (top), optimal mixed-strategy with the lowest average power that forces $R_{\min}$ (bottom).}
 \label{Fig:Equally_Spaced_Jamming_Powers}
\end{figure}

Figure \ref{Fig:Equally_Spaced_Jamming_Powers} (top) shows the expected transmission rate at the NE as a function of the jammer's average power for typical values of $P_T, N, J_{\max}$. 
The figure is sketched for four different jammers; optimal strategic jammer, semi-uniform jammer, reactive non-strategic jammer and random strategic jammer (a jammer that uses a random mixed-strategy). As it can be verified by Figure \ref{Fig:Equally_Spaced_Jamming_Powers} (and numerical simulations) that all non-optimal jamming strategies under-perform the optimal Strategic jammer.
Figure \ref{Fig:Equally_Spaced_Jamming_Powers} (bottom), shows a typical optimal mixed-strategy with the lowest average power that can force $R_{\min}$ at the NE. The randomization gain of this strategy is
% 
% In this case, the jammer's randomization gain to force the transmitter to operate at the lowest rate is
\begin{equation}
 \text{Randomization Gain}  \approx 2.9 = 4.6 \ \text{dB}
\end{equation}
% 
% 
%
% 
% ******************************************************************************************************************
% 
\section{Conclusion}
\label{Sec:Conclusion}
We formulated the interaction between an adaptive transmitter (a transmitter with multiple transmission choices) and a smart power limited jammer in a game theoretic context. We showed that packetization and adaptivity benefits a smart jammer. While the standard information-theoretic performance results for a jammed channel corresponds to pure Nash equilibrium, packetized adaptive communication leads to a lower expected game value and a mixed-strategy Nash Equilibrium. Inspired by the Shannon's capacity theorem, we defined a general utility function and a payoff matrix which may be applied to a variety of jamming problems. Furthermore, we showed the existence of optimal mixed-strategy NE for the transmitter and the jammer. We showed the existence of a threshold on jammer's average power such that if the jammer's average power exceeds this threshold then the expected value of the game at NE corresponds to the transmitter's lowest payoff; as if the jammer was using the maximum jamming power all the time. Finally, we studied an special case of optimal strategies in a discrete-time AWGN wireless channel under jamming and showed that randomization can significantly assist a smart jammer with limited average power.
%
% 
% ******************************************************************************************************************
% 
\appendices
\section{Jamming Threshold Upper Bound}
\label{Appendix_1}
% \input{Appendix_1.tex}
% n Section \ref{SubSec:Power_Threshold} we showed that the 
% \begin{equation}\label{Eq:Zi}
%  J_{\text{ave}} \geq Z_i = \frac{1}{2} J_{\max} \frac{N_T + 1}{N_T - i} \left(1 - \frac{R_{N_T}}{R_i} \right)
% \end{equation}
% where
% \begin{equation}
%  R_i = \frac{1}{2} \log \left( 1 + \frac{P_T}{N + \left( \frac{i}{N_T} \right) J_{\max}} \right)
% \end{equation}
% We stated that $Z_i$ given in \eqref{Eq:Zi} is strictly decreasing function of $i$ hence, $\max_i Z_I = Z_0 $ and an upper bound for the average power threshold is

In section \eqref{SubSec:Special_Case_II} we showed, without giving a proof, that 
\begin{equation}\label{Eq:U_I}
 \max_{0\leq i < N_T} 
 \textstyle
 \left\{ U_i \stackrel{\Delta}{=} \left( 1 - \frac{R_{N_T}}{R_i} \right) \left( \frac{1}{2} \frac{N_T + 1}{N_T - i} \right)J_{\max} \right\}=J_{\text{TH,U}}
\end{equation}
where $R_i$ and $J_{\text{TH,U}}$ are given by \eqref{Eq:Rates_forEqually_Spaced_Jamming_Powers} and \eqref{Eq:J_TH_Upper_Bound}, respectively.
% where $R_i$ is defined in \eqref{Eq:Rates_forEqually_Spaced_Jamming_Powers} and $J_{TH.U}$ is given by
% \begin{equation}\label{Eq:JTHU}
%  J_{TH,U} = U_0 = \frac{1}{2} \frac{N_T + 1}{N_T} \left(1 - \frac{R_{N_T}}{R_0} \right) J_{\max}
% \end{equation}
% i.e., $J_{TH,U}$ can be used as an approximation and an upper bound for jamming threshold when the jamming powers are equally spaced in $\big[ 0, \ J_{\max} \big]$.
\begin{proof}
To show that \eqref{Eq:U_I} is true for all $ P_T,$ $N,$  $J_{\max},$ $J_{\text{ave}} > 0$ we need to show that $U_i$ is indeed maximized for $i=0$. First, we rewrite $U_i$ as
\begin{equation}\label{Eq:Zi2}
\textstyle
  \begin{aligned}
    U_i &=  \frac{1}{2} J_{\max} \frac{N_T + 1}{N_T} \frac{1}{1-\frac{i}{N_T}}  \left(1 - \frac{R_{N_T}}{R_i} \right)  			;\quad 0 \leq i < N_T\\
	&=  \left( \frac{1}{2} J_{\max} \frac{N_T + 1}{N_T} \right) \frac{J_{\max}}{J_{\max} - (\frac{i}{N_T} J_{\max})} \left(1 - \frac{R_{N_T}}{R_i} \right)\\
  \end{aligned}
\end{equation}
define $J$ and $R(J)$ as
\begin{equation}\label{Eq:J_and_RJ}
\textstyle
  \begin{aligned}
    & J = \left( \frac{i}{N_T} J_{\max} \right)  & 0 & \leq i < N_T  \\
    & R(J) = \frac{1}{2} \log \left( 1 + \frac{P_T}{N + J} \right) & 0 & \leq J < J_{\max} \\
  \end{aligned}
\end{equation}
substituting \eqref{Eq:J_and_RJ} in \eqref{Eq:Zi2} and we have
\begin{equation}\label{Eq:ZJ}
  \begin{aligned}
    U(J) &= \left( \frac{1}{2} J_{\max} \frac{N_T + 1}{N_T} \right) \frac{J_{\max}}{J_{\max} - J} \left[1 - 			 	\frac{R(J_{\max})}{R(J)} \right]\\
	 &= a \times \frac{J_{\max}}{J_{\max} - J} \left[1 - \frac{R(J_{\max})}{R(J)} \right]\\
	 &= a\times F(J); \quad \text{where} \quad a > 0\ \text{and}\ 0\leq J < J_{\max}\\ 
  \end{aligned}
\end{equation}
If $F(J)$ in \eqref{Eq:ZJ} were a decreasing function of $J$ then $U_i$ and $U(J)$ would also be decreasing functions of $i$ and $J$ respectively. Let
\begin{equation}\label{Eq:FJ}
  F(J) = f(J)g(J) \quad \text{where} \quad 
  \begin{aligned}
    &f(J) = \frac{J_{\max}}{J_{\max} - J}\\
    &g(J) = 1 - \frac{R(J_{\max})}{R(J)}\\
  \end{aligned}
\end{equation}
For decreasing $F(J)$ we have
\begin{equation}\label{Eq:Decreasing_FJ}
 \begin{aligned}
    \frac{\partial}{\partial J}F & = g \frac{\partial}{\partial J}f + f \frac{\partial}{\partial J}g < 0\\
				 & f,g>0 \quad \text{for} \quad 0\leq J < J_{\max}  \Rightarrow \frac{\frac{\partial}{\partial J}f}{f} < - \frac{\frac{\partial}{\partial J}g}{g}\\
 \end{aligned}
\end{equation}
From \eqref{Eq:FJ} we have 
\begin{equation}\label{Eq:dg_and_df}
  \begin{aligned}
 - \frac{\frac{\partial}{\partial J}g}{g} =& 
    \frac{1}{N+J} \times \left( \frac{x}{1 + x} \right) \left( \frac{\log (1 + x_m)}{\log (1 +x)} \right) \\
      & \quad \times \left( \frac{1}{\log (1 + x) - \log (1 + x_m) } \right) \\
  \frac{\frac{\partial}{\partial J}f}{f} =& \frac{P_T^{-1} x_m x}{x - x_m}\\
  \end{aligned}
\end{equation}
where
\begin{equation}\label{Eq:x}
 x = \frac{P_T}{N + J} \ \ \text{and} \ \ x_m = \frac{P_T}{N + J_{\max}}\ \  \text{and} \
% \end{equation}
% and
% \begin{equation*}
 0 < x_m < x
\end{equation}
If we plug \eqref{Eq:dg_and_df} and \eqref{Eq:x} in \eqref{Eq:Decreasing_FJ} and simplify the resulted inequality we have
\begin{equation}
  \begin{aligned}\label{Eq:MathcalZ}
    \mathcal{Z} =& \frac{x_m^{-1}x}{1+x} \frac{x-x_m}{\log(1+x) - \log(1+x_m)} \frac{\log(1+x_m)}{\log(1+x)} \\
		 & > 1 \qquad \text{for} \quad 0 < x_m < x \\
  \end{aligned}
\end{equation}
We need to show that \eqref{Eq:MathcalZ} holds for all $ 0 < x_m < x $. Notice that 
\begin{equation}\label{Eq:MathcalZ_Limit}
  \lim _{x \rightarrow x_m^{+}} \mathcal{Z} \sim \frac{x_m^{-1} \log(1+x_m)}{x^{-1} \log(1+x)} \rightarrow 1^+ \quad \forall \ 0 < x_m < x
\end{equation}
since we have used
\begin{equation}
  \frac{d}{d z} z^{-1} \log(1+z) < 0 \quad \forall \ 0 < z
\end{equation}
and the following natural logarithm property
\begin{equation}\label{Eq:Log_Inequality}
  \frac{z}{1+z} < \log(1+z) \leq z \quad \text{for all} \quad z>0 
\end{equation}
For simplicity we rewrite inequality in \eqref{Eq:MathcalZ} as
\begin{equation}\label{Eq:MathcalZ2}
  \begin{aligned}
 & \mathcal{Z}_2 = \left[ x(x-x_m) \log(1+x_m) \right] \\ 
	       & - \left[ x_m (1+x) \log(1+x) \left( \log(1+x) - \log(1+x_m) \right) \right] 
		 > 0
  \end{aligned}
\end{equation}
As a result of \eqref{Eq:MathcalZ_Limit} we have $ \lim_{x \rightarrow x_m^{+}} \mathcal{Z}_2 \rightarrow 0^+ $ for all $0<x_m<x$. Since \eqref{Eq:MathcalZ2} holds for $x\rightarrow x_m^+ $, if $\mathcal{Z}_2$ were a strictly increasing function of $x$ for all $x>x_m$, \eqref{Eq:MathcalZ2} and \eqref{Eq:MathcalZ} would also hold as a corollary.

To show that $\mathcal{Z}_2$ is strictly increasing, we first verify that 
\begin{equation}\label{Eq:MathcalZ2_Derivative}
 \frac{\partial \mathcal{Z}_2}{\partial x} \bigg( x = x_m \bigg) = 0
\end{equation}
given that \eqref{Eq:MathcalZ2_Derivative} is true, an alternative way to proceed is to show that $\frac{\partial \mathcal{Z}_2}{\partial x}$ is itself strictly increasing function of $x$ (strictly convex function of $x$). Define $\mathcal{Z}_3$
\begin{equation}\label{Eq:MathcalZ3}
  \begin{aligned}
      & \mathcal{Z}_3 = \frac{\partial^2 \mathcal{Z}_2}{\partial x^2} \times (1+x) = 2\log(1+x_m) - 2x_m + \\ 
      & 2x \log(1+x_m)- 2x_m \log(1+x) + x_m \log(1+x_m)\\
  \end{aligned}
\end{equation}
It can be verified that for all $x>x_m$ and $x_m > 0 $ we have  $ \lim_{x \rightarrow x_m^{+}} \mathcal{Z}_3 > 0 $. Taking the partial derivate of $\mathcal{Z}_3$ with respect to $x$ and we have
\begin{equation}
      \frac{\partial \mathcal{Z}_3}{\partial x} = 2 \left[ \log(1+x_m) - \frac{x_m}{1+x} \right]
\end{equation}
but from \eqref{Eq:Log_Inequality} we have
\begin{equation}
  \begin{aligned}
      &\log(1+x_m) > \frac{x_m}{1+x_m}>\frac{x_m}{1+x} \quad \text{for all} \quad x>x_m\\
      & \Rightarrow \quad 2 \left[ \log(1+x_m) - \frac{x_m}{1+x} \right] >0 \quad \forall x>x_m>0\\
  \end{aligned}
\end{equation}
and hence we have 
\begin{equation}
 \frac{\partial \mathcal{Z}_3}{\partial x} > 0 \qquad \text{for all} \quad x>x_m>0
\end{equation}
Consequently, $\mathcal{Z}_2$ is indeed an increasing function of $x$ for all $0<x_m<x$. Taking the reverse steps that resulted in \eqref{Eq:MathcalZ2} and \eqref{Eq:MathcalZ} we can conclude that $U_i$ in \eqref{Eq:U_I} is indeed a strictly decreasing function of $i$ and hence it is maximized for $i=0$.
% $J_{TH,U}$ given in \eqref{Eq:JTHU} is an upper bound for $J_{TH}$.
\end{proof}
%
% 
% ******************************************************************************************************************
% 
\section{Linear Programming and Constrained Two-Player Zero-Sum Games}
\label{Appendix_2}
Consider a two-player zero-sum game in which, due to some practical reason, not all mixed-strategies are feasible strategies \cite{Owen}. Assume the mixed-strategies $\boldsymbol{x} \in \mathbf{X}$  and $\boldsymbol{y} \in \mathbf{Y}$ (player I's and player II's mixed-strategies respectively) must be chosen from some hyper-polyhedron, i.e., from constraint sets defined by linear inequalities and equalities.
If we let $Z$ be the game matrix, player I's problem is to find
\begin{equation}\label{Eq:Player_I_Problem}
 \max_{\boldsymbol{x} \in \mathbf{X}} \left( \min_{\boldsymbol{y} \in \mathbf{Y}} \ \boldsymbol{x} Z \boldsymbol{y}^T \right)
\end{equation}
where $\boldsymbol{x}$ and $\boldsymbol{y}$ are row vectors and the sets $\mathbf{X}$ and $\mathbf{Y}$ in the most general case, are defined by 
\begin{equation}
\label{Eq:Constrained_Sets_General_Case}
\textstyle
  \mathbf{X}: \ \begin{cases}
                \boldsymbol{x} B \leq \boldsymbol{c}\\
		\boldsymbol{x} \geq \boldsymbol{0}
               \end{cases}
  \qquad
  \mathbf{Y}: \ \begin{cases}
                \boldsymbol{y} E^T \geq \boldsymbol{f}\\
		\boldsymbol{y} \geq \boldsymbol{0}
               \end{cases}
\end{equation}
Similarly, player II's problem is to find
\begin{equation}\label{Eq:Player_II_Problem}
 \min_{\boldsymbol{y} \in \mathbf{Y}} \left( \max_{\boldsymbol{x} \in \mathbf{X}} \ \boldsymbol{x} Z \boldsymbol{y}^T \right)
\end{equation}
Consider the optimization problem in \eqref{Eq:Player_I_Problem}, the minimization problem inside the parenthesis can be represented by a linear program whose objective function depends on $\boldsymbol{x}$. From the duality theorem \cite{Owen} and if the program is feasible and bounded, then the two programs
\begin{equation}
\textstyle
 \text{Minimize} \ \left( \boldsymbol{x} Z \right) \boldsymbol{y}^T \quad \text{subject to:} \quad
	       \begin{cases}
                \boldsymbol{y} E^T \geq \boldsymbol{f}\\
		\boldsymbol{y} \geq \boldsymbol{0}
               \end{cases}
\end{equation}
and 
\begin{equation}
\textstyle
\label{Eq:Dual_Program}
 \text{Maximize} \  \boldsymbol{z} \boldsymbol{f}^T  \quad \ \ \ \ \text{subject to:} \quad 
	       \begin{cases}
                \boldsymbol{z} E \leq \boldsymbol{x}Z\\
		\boldsymbol{z} \geq \boldsymbol{0}
               \end{cases}
\end{equation}
will have the same value (where $\boldsymbol{z}$ is an auxiliary variable). If we plug the dual program \eqref{Eq:Dual_Program} in \eqref{Eq:Player_I_Problem}, player I's problem becomes a pure maximization problem, i.e
\begin{equation}
\label{Eq:Player_I_Maximization_Problem}
\textstyle
 \text{Maximize} \ \boldsymbol{z} \boldsymbol{f}^T \quad  \ \text{subject to:} \quad
		\begin{cases}
		 \boldsymbol{z} E - \boldsymbol{x} Z \leq \boldsymbol{0}\\
		 \boldsymbol{x} B \leq \boldsymbol{c} \\
		 \boldsymbol{x}, \boldsymbol{z} \geq \boldsymbol{0}\\
		\end{cases}
\end{equation}
This problem can be solved by the usual linear program algorithms i.e, simplex algorithm. In a similar way, it can be shown that player II's problem could be reduced to a pure minimization problem, i.e.
\begin{equation}
\textstyle
\label{Eq:Player_II_Minimization_Problem}
 \text{Minimize} \ \boldsymbol{c} \boldsymbol{s}^T \quad  \ \text{subject to:} \quad
		\begin{cases}
		 \boldsymbol{s} B^T - \boldsymbol{y} Z^T \geq \boldsymbol{0}\\
		 \boldsymbol{y} E^T \geq \boldsymbol{f} \\
		 \boldsymbol{y}, \boldsymbol{s} \geq \boldsymbol{0}\\
		\end{cases}
\end{equation}
where $\boldsymbol{s}$ is an auxiliary variable. It can be verified that the program \eqref{Eq:Player_II_Minimization_Problem} is the dual of the program \eqref{Eq:Player_I_Maximization_Problem} and therefore, if both are feasible and bounded, they will have the same value and the constrained game will have a NE in mixed-strategies.

By using appropriate set of matrices and vectors, we can reformulate Transmitter's and jammer's strategy constraint sets  defined in Section \ref{SubSec:Mixed_Strategy_Set} into the general format introduced in \eqref{Eq:Constrained_Sets_General_Case}. Specifically, consider the transmitter's constraint set, the following set of matrix, $B$, and vector, $\boldsymbol{c}$,  can be used to represent transmitter's constraint set;
\begin{equation}
\label{Eq:B_Matrix}
 B_{(N_T+1)\times 2} = \left[ \begin{array}{c|c}  
			\boldsymbol{1}^{T}_{(N_T+1)\times 1} & -\boldsymbol{1}^{T}_{(N_T+1)\times 1} 
			\end{array} \right]
\end{equation}
\begin{equation}
 \boldsymbol{c}_{1\times2} =  \begin{bmatrix}
                               1 & -1
                              \end{bmatrix}
\end{equation}
similarly, jammer's constraint set can be represented by the following matrix, $E$, and vector, $\boldsymbol{c}$,
\begin{multline}
 E^{T}_{(N_T+1)\times 3} = \\  \left[ \begin{array}{c|c|c}  
			    \boldsymbol{1}^{T}_{(N_T+1)\times 1}  & -\boldsymbol{1}^{T}_{(N_T+1)\times 1}  & -J^T_{(N_T+1)\times 1}  \end{array} \right] 
\end{multline}
\begin{equation}
\label{Eq:f_Vector}
 \boldsymbol{f}_{1\times3} = \begin{bmatrix}
                   1 & -1 & -J_{\text{ave}} \\
                  \end{bmatrix}
\end{equation}
From \eqref{Eq:Player_I_Maximization_Problem} and \eqref{Eq:f_Vector} the maximization program objective function becomes
\begin{equation}
\label{Eq:Objective_Function}
 \boldsymbol{z}\cdot \boldsymbol{f}^T = z_1 - z_2 - J_{\text{ave}}z_3
\end{equation}

To show that the maximization program is bounded and feasible and hence has a solution in mixed-strategies, we need to show that the objective function, given in \eqref{Eq:Objective_Function}, is bounded and the constraint set defined by the set of matrices and vectors in \eqref{Eq:B_Matrix}-\eqref{Eq:f_Vector} is non-empty. Assume, for now, that the objective function is unbounded, we must have 
\begin{equation}
\label{Eq:Objective_Function_Assumption}
 z_1 - z_2 - J_{\text{ave}}z_3 > Z_{\max} = \max_{i,j} Z_{ij}
\end{equation}
for some $\boldsymbol{z} = [ z_1\ z_2\ z_3] \geq 0$ that satisfy the constraints in \eqref{Eq:Player_I_Maximization_Problem}. Consider the first inequality in \eqref{Eq:Player_I_Maximization_Problem}, multiplying vectors $\boldsymbol{z}$ and $\boldsymbol{x}$ by the first column of matrices $E$ and $Z$ results in
\begin{align}
\label{Eq:Z_Max_Assumption}
  & z_1 - z_2 - J_0z_3 - \boldsymbol{x} Z_{:,1} \leq 0 \notag \\
  & \Rightarrow \quad z_1 - z_2 \leq \boldsymbol{x}Z_{:,1} \leq Z_{\max}=\max_{i} Z_{i,1} \notag \\
  & \Rightarrow \quad Z_{\max} \geq z_1 - z_2
\end{align}
where $Z_{:,1}$ denotes the first column of $Z$ and by assumption we have $J_0 = 0$. Plugging \eqref{Eq:Z_Max_Assumption} in \eqref{Eq:Objective_Function_Assumption} results in $z_3 < 0$ which is in contradiction with $\boldsymbol{z} > 0$. As a result, the objective function cannot be greater than $Z_{\max}$ and hence the program is bounded. To show that the constraint set is non-empty, we need to show that there exists at least one pair of $(\boldsymbol{x}, \boldsymbol{z})$ that satisfies the constraints in \eqref{Eq:Player_I_Maximization_Problem}. From the first inequality in \eqref{Eq:Player_I_Maximization_Problem} we must have
\begin{equation}
\label{Eq:Feasibily_Proof_1}
 z_1 - z_2 - J_iz_3 \leq \boldsymbol{x} Z_{:,i} \quad \forall  i: 1\leq i \leq N_T+1
\end{equation}
but $0\leq \boldsymbol{x} Z_{:,i} \leq Z_{\max}$ for $1\leq i \leq N_T+1$ and for all probability vectors $\boldsymbol{x}$, therefore if we let 
\begin{equation}
\label{Eq:Feasibily_Proof_2}
 z_1 -z_2 - J_iz_3 \leq 0 
\end{equation}
then \eqref{Eq:Feasibily_Proof_1} would be satisfied for all $1\leq i \leq N_T+1$. Choosing $z_3 = 0$ and letting $0< z_1 < z_2$ satisfies the inequality \eqref{Eq:Feasibily_Proof_2} and as a result the transmitter's constraint set is nonempty. 

Therefore, the transmitter's maximization program is feasible and bounded and has a solution. As a result of the duality theorem, the dual of this program, jammer's minimization problem, is also feasible and bounded has the same solution (NE in mixed-strategies).
\end{document}